\documentclass[runningheads]{llncs}
\usepackage[T1]{fontenc}
\usepackage{amsmath,amsfonts}
\usepackage{algorithmic}
\usepackage{algorithm}
\usepackage{xspace}
\newcommand{\HwPTF}{HW-PTF}

\newcommand{\OTHERD}{\mbox{\textsc{M2D}}\xspace}
\newcommand{\HALFD}{\mbox{\textsc{M2H$_d$}}\xspace}
\newcommand{\ATDIST}{\mbox{A{\scriptsize{T-}}D}\xspace}

\newcommand{\STAY}{\mbox{\textsc{Stay}}\xspace}
\newcommand{\OTHER}{\mbox{\textsc{M2O}}\xspace}
\newcommand{\HALF}{\mbox{\textsc{M2H}}\xspace}

\newcommand{\BLACK}{\mbox{\textsc{Black}}\xspace}
\newcommand{\WHITE}{\mbox{\textsc{White}}\xspace}

\newtheorem{observation}{Observation}

\usepackage{geometry}
\usepackage{graphicx}
\begin{document}
\title{Hardwired Pattern Formation by Mobile Robots with Common Unit Distance} 
%
%

\author{Yuta Kojima\inst{1} \and
S\'{e}bastien Tixeuil\inst{2}\orcidID{0000-0002-0948-7172} \and
Yukiko Yamauchi\inst{3}\orcidID{0009-0009-8459-6676}}

\authorrunning{Y. Kojima et al.}
%

\institute{Kyushu University, Japan \email{kojima.yuta.457@s.kyushu-u.ac.jp} \\ 
Sorbonne Universit\'{e}, CNRS, LIP6, IUF, France \email{Sebastien.Tixeuil@lip6.fr} \\ 
Kyushu University, Japan \email{yamauchi@inf.kyushu-u.ac.jp}}
\maketitle              
\begin{abstract}
The \emph{pattern formation (PTF) problem} requires 
mobile robots to form a specified target pattern. 
Existing papers investigated the PTF problem and 
revealed the effect of obliviousness and synchronization on 
distributed coordination of mobile robots. 
However, the PTF problem allows translation, rotation, 
and scaling of the target pattern. 
In this paper, we introduce a novel pattern formation problem, 
called the \emph{hardwired pattern formation (\HwPTF) problem} 
that requires the robots to form a given target pattern 
in a specified size. 
Although two oblivious semi-synchronous robots 
cannot solve the \HwPTF \ problem of multiplicity two 
(i.e., the rendezvous problem), 
we show that they can solve the \HwPTF \ problem without multiplicity. 
We also show that two oblivious asynchronous robots equipped with lights 
can solve the \HwPTF \ problem, 
while oblivious asynchronous robots cannot.  
We finally present a size-adjusting algorithm for 
more than four oblivious semi-synchronous robots, 
that yields a \HwPTF \ algorithm 
when combined with some existing pattern formation algorithms.

\end{abstract}

\section{Introduction} 
The \emph{pattern formation (PTF) problem} has been one of the most important problems 
in distributed coordination of \emph{mobile robots}. 
The problem requires the robots to form a target pattern 
without any centralized control. 
Each robot is an autonomous mobile computing entity, 
that observes the positions of other robots, 
computes its next position, and 
moves toward the next position. 
The robots are \emph{anonymous} and \emph{uniform} 
in the sense that the robots are indistinguishable points 
and computes its next position by a common deterministic algorithm. 
The robots are not equipped with a common coordinate system 
nor memory to store past observations nor computation, i.e., 
they are \emph{disoriented} and \emph{oblivious}. 
Each robot performs a unit of action when it is activated 
and performs observation, computation, and movement. 
We consider three types of activation schedules, called 
the \emph{fully-synchronous} (FSYNC) model, the \emph{semi-synchronous} (SSYNC) model, 
and the \emph{asynchronous} (ASYNC) model. 
Because the robots do not have a common coordinate system, 
the pattern formation problem allows 
translation, rotation, and scaling of the target pattern. 
Existing literature showed 
the class of formable patterns and 
the effect of obliviousness and synchronization~\cite{FYOKY15,SY99,YS10}, 
and these results are further extended from 2D space to 3D space~\cite{YUKY17}. 
For example, Suzuki and Yamashita introduced the notion of \emph{symmetricity} 
of a set $P$ of points, 
which is essentially the rotational symmetry of $P$. 
Then they showed that the robots can form a target pattern $F$ 
from an initial configuration $I$ if and only if $\rho(I)$ divides $\rho(F)$ 
irrespective of obliviousness and synchronization. 
That is, the robots cannot resolve their initial rotational symmetry 
in their positions and local coordinate systems. 
However, there exists an exceptional case for the 
\emph{rendezvous problem}, that requires two robots to form a single 
point of multiplicity two. 
When the robots are fully-synchronized, 
they can solve the rendezvous problem by moving to the midpoint. 
Suzuki and Yamashita showed that the two robots cannot solve the 
rendezvous problem when the robots are SSYNC or ASYNC~\cite{SY99}. 
Hence the rendezvous problem separates computational power 
of FSYNC robots from that of SSYNC or ASYNC robots. 

In this paper, we introduce a novel pattern formation problem, 
called the \emph{hardwired pattern formation (\HwPTF) problem}, 
that requires the robots to form a given target pattern 
in a specified size. 
That is, the robots are required to reach a configuration 
congruent to the target pattern. 
We add the minimum requirement of a common unit distance 
for the robots to agree on the specified size. 
While the pattern formation problem has shed light to theoretical aspect 
of distributed coordination for mobile robots, 
its application to hardware robots, drones, and other mobile devices is limited. 
For example, consider surveillance of a large field by mobile sensors. 
Existing pattern formation algorithms~\cite{FYOKY15,SY99,YS10} 
enable the mobile sensors to form an appropriate pattern to cover the field, 
however its size may disrupt communication connectivity or sensing coverage. 
Size specification has many applications including 
engineering, nanomanufacturing, structural biology, and so on. 

\noindent{\bf Our contribution.} 
We show that the \HwPTF \ problem is 
both practically and theoretically important.  
First, we show that two SSYNC robots can solve 
the \HwPTF \ problem except the rendezvous problem. 
That is, two robots can stop at specified distance $d(\neq 0)$ from an 
arbitrary initial configuration without multiplicity. 
Our key idea is the \emph{Z-move} that enables two robots finish the \HwPTF \ 
under any SSYNC activation. 
We then show that two ASYNC robots cannot solve 
the \HwPTF \ problem for sufficiently large $d$. 
The Z-move plays an important role to show the impossibility 
because any SSYNC execution is also an ASNYC execution. 
That is, two ASYNC robots need the Z-move to finish \HwPTF. 
We will show that there exists an ASYNC activation 
that prevents two robots from solving the \HwPTF \ problem forever. 
These results also demonstrate that the \HwPTF \ problem for two robots 
separates computational power 
of oblivious SSYNC robots from that of oblivious ASYNC robots. 
We then show that oblivious ASYNC robots equipped with \emph{lights} with two colors 
can solve the \HwPTF \ problem. 
Our key technique is to establish an isomorphism between 
an ASYNC execution of an rendezvous algorithm such as~\cite{Her2Cols} 
and an ASYNC execution of a \HwPTF \ process. 
Finally, we present a \HwPTF \ algorithm 
for more than four oblivious SSYNC robots. 
Specifically, we present a size-adjusting algorithm, 
which yields a \HwPTF \ algorithms 
when combined with some existing pattern formation algorithms~\cite{FYOKY15,SY99,YS10}. 

\noindent{\bf Related work.} 
Many papers investigate the effect of additional capabilities or restrictions 
on the formable patterns. 
Flocchini et al. showed that odd number of oblivious ASYNC  
robots can form 
an \emph{arbitrary pattern} when they agree on the direction and orientation 
of one axis even when they have no common chirality~\cite{FPSW08}. 
Dieudonn\'{e} et al. showed that oblivious ASYNC  
robots can form an arbitrary pattern when there exists a 
designated leader robot~\cite{DPV10}. 
Cieliebak showed that more than two oblivious ASYNC robots 
can solve the \emph{point formation problem}~\cite{CFPS12} 
while two oblivious SSYNC robots cannot~\cite{SY99}. 
Interestingly, if robots do \emph{not} agree on the unit distance, SSYNC rendezvous and gathering becomes feasible~\cite{Bramas23} even if the starting configuration is bivalent (that is, with two location hosting an equal number of robots).

While the effect of common directions or orientation on 
the pattern formation problem has been investigated, 
few papers consider the effect of common distance. 
Several papers consider the pattern formation problem 
for robots with limited visibility. 
Yamauchi et al. showed that 
non-oblivious SSYNC robots with limited visibility 
can form a ``small'' target pattern under the symmetricity condition. 
The trick is first gather the robots within their visibility range and 
then make them form the target pattern~\cite{YY13}. 
Hahn et al. showed that 
oblivious FSYNC robots with limited visibility 
can form a ``large'' target pattern under the symmetricity condition 
and additional connectivity condition~\cite{HHK24}. 
However, in these papers, 
each robot does not know the size of its visibility. 

To the best of our knowledge, our result is the first result 
on the effect of a common unit distance to the pattern formation problem.

\section{Preliminary}

\subsection{Robot System}
We consider a set $R = \{r_1, r_2, \ldots, r_n\}$ of $n$ point robots. 
The robots are indistinguishable, and we use $r_i$ just for notation. 
Let $p_i(t)$ be the position of $r_i$ at time $t$ in the global coordinate system $Z_0$.  
The \emph{configuration} of $R$ at time $t$ 
is the multiset $P(t) =\{p_i(t) \mid i = 1, 2, \ldots, n\}$. 
We use $dist(p,q)$ to describe the distance between two points $p$ and $q$ in $Z_0$.

The robots repeat a \emph{Look-Compute-Move} cycle, 
that consists of the following three phases. 
In the Look phase, robot $r_i$ obtains a snapshot of the robot system, i.e.,  
the positions of all robots observed in its local coordinate system $Z_i$. 
All global and local coordinate systems are right-handed $x$-$y$ coordinate system 
with a common unit distance. 
For each robot $r_i$, the directions and orientation of $Z_i$ never change, 
but the origin of $Z_i$ moves according to the movement of $r_i$, 
that is, the origin is the current position of $r_i$. 
Hence, the local coordinate system must be denoted by $Z_i(t)$, 
but we use $Z_i$ when it is clear by the context. 
We assume that the robots are equipped with 
a \emph{weak multiplicity detection capability}, 
i.e., a robot can determine whether its current position is occupied 
by other robots or not.\footnote{We do not use multiplicity detection capability 
because our proposed algorithms does not make any multiplicity 
and we refer to multiplicity in only impossibility results.}
In the Compute phase, robot $r_i$ computes its next position 
by a common deterministic algorithm $A$. 
The sole input to $A$ is the observation obtained in the preceding Look phase, 
i.e., we consider \emph{oblivious} robots. 
The output of $A$ is the coordinates of the next position of $r_i$ in $Z_i$. 
In the Move phase, robot $r_i$ moves toward its next position, 
however, $r_i$ may stop en route after it moves by 
the \emph{minimum moving distance} $\delta$, which is not given to the robots, 
i.e., we consider \emph{non-rigid movement}.\footnote{
Some of our impossibility result involve \emph{rigid movement}, 
that allows robots always reach their next positions.} 

We consider two types of schedulers (i.e., synchronization models) for mobile robots. 
In the \emph{semi-synchronous} (SSYNC) model, 
we consider discrete time $t = 0,1,2, \ldots$ 
and non-empty subset of the robots are \emph{activated} at each time step. 
The active robots execute a Look-Compute-Move cycle synchronously 
with each of the three phases completely synchronized. 
Hence, the configuration of the robots changes from $P(t)$ to 
$P(t+1)$ according to the common algorithm $A$. 
We call the evolution $P(0), P(1), P(2), \ldots$ of the robot system $R$
an \emph{execution} of algorithm $A$. 
There exists multiple executions starting from $P(0)$ 
because of the activation schedule and non-rigid movement. 
In the \emph{asynchronous} (ASYNC) model, robots execute their Look-Compute-Move cycle independently. The length of each Look-Compute-Move cycle is finite but arbitrary, 
and a robot can be observed while moving. 
Let $t_0, t_1, t_2, \ldots, $ be the time instance where at least one robot 
obtains a snapshot in a Look phase. 
We consider $P(t_0), P(t_1), P(t_2), \ldots$ as the execution 
$P(0), P(1), P(2), \ldots$ in the ASYNC model. 

For a given set of points $P$, 
let $SEC(P)$ and $c(P)$ denote the \emph{smallest enclosing circle} (SEC) of $P$ 
and its center, respectively. 
We use $rad(P)$ to represent the radius of $SEC(P)$. 
We call the largest circle centered at $c(P)$ and contains no point of $P$ 
in its interior the \emph{largest empty circle} (LEC) of $P$, 
denoted by $LEC(P)$. 
By definition, $LEC(P)$ contains at least one point of $P$ on its boundary. 
Let $d_{min}(P)$ denote the minimum distance between the points of $P$. 

\subsection{Hardwired Pattern Formation Problem} 
The \emph{hardwired pattern formation} problem requires the robots 
to form a given target pattern in a given size 
when the robots agree on the unit distance. 
The input to the problem is a multiset $F$ of coordinates of $n$ points in $Z_0$ 
and the \emph{size} $d$ of the target pattern, which is $2 \cdot rad(F)$, i.e., 
the diameter of $SEC(F)$. 
We use \HwPTF$(n,F,d)$ to describe the hardwired pattern formation problem 
for $n$ robots, target pattern $F$, and its size $d$. 

An algorithm \emph{solves} the \HwPTF$(n,F, d)$ 
if any execution $P_0, P_1, P_2, \ldots$ has a finite $t$ 
that satisfies $P_{t+i} = Z'(F)$ for $i=1, 2, \ldots$ 
where $Z'$ is obtained by a rotation and/or translation, but not scaling 
on the global coordinate system $Z_0$. 
That is, $P_{t+i}$ is congruent to $F$. 
We assume that the initial configuration $P(0)$ does not contain 
any multiplicity. 

\subsection{Symmetricity of a Set of Points}
We briefly introduce the symmetricity of the set of points 
introduced by \cite{FYOKY15}. 
Given a set $P$ of points, we consider the decomposition of $P$ into 
regular $k$-gons centered at $c(P)$. 
Here, we say a set of $k$ points is a regular $k$-gon 
if they are placed at the vertices of some regular $k$-gon. 
A point is a regular $1$-gon with an arbitrary center and 
two points is a regular $2$-gon. 
The \emph{symmetricity} of $P$ is the maximum number of such $k$ 
with the following exception; 
$\rho(P) = 1$ when $c(P) \in P$. 
We show some examples of symmetricity: 
\begin{itemize}
    \item When $P$ forms a regular $n$-gon ($|P|=n$), $\rho(P)$ is $n$.  
    \item When $P$ forms a square $\rho(P) = 4$. 
    \item When $P$ forms a non-square rectangle, $\rho(P)=2$. 
\end{itemize}

The \emph{$\rho(P)$-decomposition} of $P$ is 
a decomposition of $P$ into subsets, say $\{P_1, P_2, \ldots, P_k\}$, 
where each $P_i$ forms a regular $\rho(P)$-gon centered at $c(P)$. 
Thus, $k = n/\rho(P)$. 
Suzuki and Yamashita showed that the robots can agree on the total ordering 
of $\{P_1, P_2, \ldots, P_{n /\rho(P)}\}$ by introducing 
appropriate local view and ordering among them~\cite{SY99,YS10}. 
We adopt a total ordering of the elements of 
the \emph{$\rho(P)$-decomposition} of $P$
so that $P_1$ is on $LEC(P)$ and $P_{n/\rho(P)}$ is on $SEC(P)$, 
i.e., the distance from $c(P)$ is considered or added as the first criteria for ordering. 

\section{Hardwired Pattern Formation for Two Robots}

In this section, we consider the \HwPTF \ problem 
for two robots. 
The target patterns of \HwPTF$(2,F,d)$ is 
either a point of multiplicity two (i.e., $d=0$) 
or two points (i.e., $d > 0$). 
Hence, we use \HwPTF$(2,d)$ instead of \HwPTF$(2,F,d)$.

\subsection{Two Oblivious SSYNC Robots} 
\label{subsec:SSYNC-2}

When $d=0$, \HwPTF$(2,d)$ degenerates into the rendezvous problem 
with a common unit distance. 
Suzuki and Yamashita showed that the rendezvous problem is not solvable by oblivious deterministic SSYNC robots without a common unit distance~\cite{SY99}. 
The impossibility result was later generalized by Courtieu et al~\cite{Courtieu15} to an even number of robots, with the ability to detect multiplicity, and the same unit distance for all robots. 
Hence, the following observation holds:
\begin{observation}
Two oblivious deterministic SSYNC robots with rigid movement cannot solve \HwPTF$(2,d)$ when $d=0$. 
\end{observation}

In the following, we will show 
\HwPTF$(2,d)$ is solvable when $d >0$ by a deterministic algorithm.

\begin{theorem}
\label{thm:HWPTF-2d}
Two oblivious deterministic SSYNC robots with non-rigid movement can solve \HwPTF$(2, d)$ when $d > 0$. 
\end{theorem}

We present a \HwPTF \ algorithm for two oblivious deterministic SSYNC robots. 
First, consider the simple case of an initial configuration where the inter-robot distance is $3d$, assuming collision-less and rigid movements. 
Then, the following algorithm can solve \HwPTF$(2, d)$: 
An active robot moves toward the other robot by distance $2d$. 
If only one of the robots is activated at time $0$, 
rigid movement guarantees that the active robot approaches the other robot to distance $d$ and stops. 
If both robots are activated at time $0$, rigid movement makes the two robots pass each other and stop at distance $d$. 
That is, from this initial configuration where robots are $3d$ apart, there exists an algorithm for the two robots to solve HW-PTF($d$) in one SSYNC step. 

We then consider a more general initial configuration, where the inter-robot distance is $a \neq d$. 
Specifically, we consider an initial configuration, 
where robot $r_1$ is at $p=(-a/2, 0)$ and $r_2$ is at $q=(a/2, 0)$ in $Z_0$. 
Let $p'$ and $q'$ be the next positions for $r_1$ and $r_2$ to 
finish \HwPTF$(2, d)$ in one SSYNC step with rigid movement. 
Hence, $p'$ and $q'$ must satisfy the following three equations: 
\begin{eqnarray*}
dist(p',q') &=& d \\ 
dist(p',q) &=& d \\ 
dist(p,q') &=& d. 
\end{eqnarray*}
Additionally, assume $p'$ and $q'$ have symmetric positions with respect to $(0,0)$, since robots may have the same observation due to their inconsistent local coordinate systems. 
Hence, if $p'=(x, y)$, then $q'=(-x, -y)$. 

By the above discussion, $x$ and $y$ must satisfy the following equations: 
\begin{eqnarray*}
2 \sqrt{x^2 + y^2} &=& d \\ 
(x-a/2)^2 + y^2 &=& d^2 
\end{eqnarray*}
Thus, we have 
\begin{eqnarray*}
x &=& \frac{a^2 - 3d^2}{4a} \\ 
y &=& \pm \frac{\sqrt{(a+d)(a-d)(3d-a)(3d+a)}}{4a}.
\end{eqnarray*}
Additionally, we have $a/2 - d \leq x \leq a/2 + d$ 
because $p'$ is on the circle centered at $q$. 
Thus, when $d \leq a \leq 3d$, we have points $p'$ and $q'$ that satisfy the three equations. 

For simplicity, let $y = \frac{\sqrt{(a+d)(a-d)(3d-a)(3d+a)}}{4a}$. 
Then, the four points $p$, $p'$, $q'$, and $q$ forms Z, and we call this movement the \emph{Z-move} (Figure~\ref{fig:z-move}). 
 \begin{figure}[t]
  \centering
  \includegraphics[width=8cm]{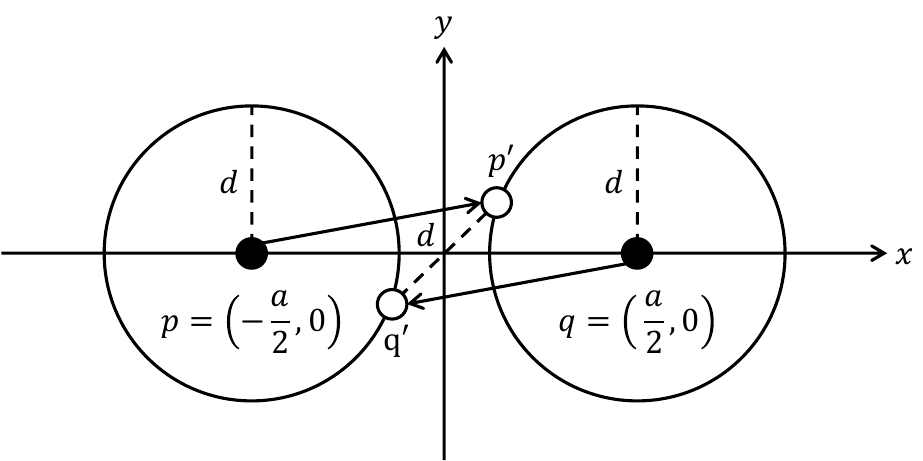} 
  \caption{Z-move} 
  \label{fig:z-move} 
\end{figure}

Algorithm~\ref{alg:HWPTF-2} shows our proposed algorithm for \HwPTF$(2, d)$. 
The algorithm consists of three moves, 
the Z-move, the expanding move, and the shrinking move. 
Let $a$ be the distance between the two robots. 
When $d < a \leq 3d$, Algorithm~\ref{alg:HWPTF-2} 
makes the robot perform the Z-move. 
When $a <d $, Algorithm~\ref{alg:HWPTF-2} 
makes the robot move distance $d$ away from the current position of the other robot 
so that the inter-robot distance is at least $d$ in a resulting configuration. 
When $a > 3d$, Algorithm~\ref{alg:HWPTF-2} 
makes the robot move to a point at distance $d/2$ from the 
midpoint of the current positions of the two robots.

\begin{algorithm}
\caption{\textsc{Form\_Hw\_Pattern\_For\_Two\_Robots$(P, d)$ at robot $r_i$}}
\label{alg:HWPTF-2}
\begin{tabbing}
xxx \= xxx \= xxx \= xxx \= xxx \= xxx \= xxx \= xxx \= xxx \= xxx
\kill 
{\bf Input by the current observation $P$} \\ 
\> $r_j$: the other robot \\
\> $p$: current position of $r_i$ \\
\> $q$: current position of $r_j$ \\
\> $a$: distance between $r_i$ and $r_j$ \\ 
\> Let $Z^*_0$ be the right-handed $x$-$y$ coordinate system s.t. $p=(-a/2,0)$ and $q = (a/2, 0)$. \\ 
\\ 
{\bf Algorithm} \\ 
\> {\bf If} $d<a<3d$ \\ 
\> \> // Z-move \\ 
\> \> $p'= \left( \frac{a^2-3d^2}{4a},-\frac{\sqrt{(a+d)(a-d)(3d-a)(3d+a)}}{4a} \right)$ \\ 
\> \> {\bf Return} $p'$ in $Z_i$ \\ 
\> {\bf else if} $a<d$ \\ 
\> \> // Expanding move \\ 
\> \> $p'=(-(d-\frac{a}{2}),0)$ \\ 
\> \> {\bf Return} $p'$ in $Z_i$ \\ 
\> \> {\bf else if} $3d\leq a$ \\ 
\> \> \> // Shrinking move \\ 
\> \> \> $p'=(-\frac{d}{2},0)$ \\ 
\> \> \> {\bf return} $p'$ in $Z_i$ \\ 
\> \> {\bf end if} \\ 
\> {\bf end if} 
\end{tabbing}
\end{algorithm}

We then show the correctness of Algorithm~\ref{alg:HWPTF-2} 
for two oblivious SSYNC robots with chirality and non-rigid movement. 
We need to consider the non-rigid movement in addition to 
the Z-move, which circumvent the SSYNC activation timing.  
For example, when the robots stop en route during the Z-move, 
their distance may become smaller than $d$, 
because the distance between the trajectories of the two robots 
(i.e., line segments $pp'$ and $qq'$) can be smaller than $d$. 
Hence, the two robots perform the expanding move 
until their distance becomes larger than $d$. 
Then, the two robots execute the Z-move again. 
We will show when we focus on time steps, 
where some robot executes the Z-move, 
the inter-robot distance gradually decreases 
and eventually the moving distance of the Z-move 
becomes smaller than the minimum moving distance $\delta$. 
Hence, an activated robot(s) can reach the next point of the Z-move 
and the two robots finish \HwPTF$(2,d)$. 

For a given execution $P_0, P_1, \ldots$ of Algorithm~\ref{alg:HWPTF-2}, 
we focus on the sequence $d_0, d_1, \ldots$ of distance between the two robots, 
where $d_t = dist(p_i(t), p_j(t))$. 

We first show that starting from an arbitrary initial configuration, 
the two robots eventually execute a Z-move. 
We have the following two lemmas. 

\begin{lemma}
\label{lemma:expand}
When $d_t < d$, 
there exists finite $t' > t$ such that $d \leq d_{t'} \leq 3d$. 
\end{lemma}
\begin{proof}
When $d_t < d$, at least one robot executes the expanding move in $P(t)$. 
The expanding move moves the two robots on the line containing $p_i(t)$ and $p_j(t)$, 
and each moving robot moves away from the other robot. 
Thus, we have $d_{t+1} - d_t \geq \delta$. 
We also have $d_{t+1} < 2d$ because $d_{t+1}$ takes the largest value 
when the two robots move to their next positions in $P(t)$. 
When $d_{t+1} < d$, at least one of the two robots perform the expanding move. 
Thus, we have $d_t < d_{t+1} < \ldots$ and 
after at most $\lceil (d - d_t)/\delta \rceil$ steps, 
the inter-robot distance becomes larger than $d$. 
\qed
\end{proof}

\begin{lemma}
\label{lemma:shrink}
When $d_t \geq 3d$, 
there exists finite $t' > t$ such that $d \leq d_{t'} \leq 3d$. 
\end{lemma}
\begin{proof}
When $d_t \geq 3d$, at least one robot executes the shrinking move in $P(t)$. 
The shrinking move moves the two robots on the line segment $p_i(t) p_j(t)$, 
and each moving robot moves toward the other robot. 
Thus, we have $d_t - d_{t+1} \geq \delta$.
We also have $d_{t+1} \geq d$ because $d_{t+1}$ takes the smallest value 
when the two robots move to their next positions in $P(t)$. 
When $d_{t+1} \geq 3d$, at least one of the two robots perform the shrinking move. 
Thus, we have $d_t > d_{t+1} > \ldots$ and 
after at most $\lceil (d_t - 3d)/\delta \rceil$ steps, 
the inter-robot distance becomes smaller than $3d$. 
\qed
\end{proof}

We then show that the two robots eventually solve \HwPTF$(2,d)$. 
\begin{lemma}
\label{lemma:zmove}
When $d < d_t < 3d$, 
there exists finite $t' > t$ such that $d_{t'} = d$. 
\end{lemma}
\begin{proof}
Let $p_i(t) = p$, $p_j(t)=q$ in $P(t)$. 
If $r_i$ and/or $r_j$ are activated in $P(t)$, they perform the Z-move. 
Let $p'$ and $q'$ be the next positions of $r_i$'s Z-move 
and $r_j$'s next move, respectively. 
The moving distance of $r_i$ is 
\begin{equation}\label{equ:zdist}
dist(p,p') = \frac{1}{2}\sqrt{d_t^2-d^2}, 
\end{equation}
which becomes smaller as $d_t$ becomes smaller. 

Due to non-rigid movement, robots may stop en route when 
they perform a Z-move. 
In this case, the inter-robot distance may become smaller than $d$ 
and (at least) one robot performs the expanding move in $P(t+1)$. 
That is, an execution $P(t), P(t+1), P(t+2), \ldots$ 
may contain the expanding move. 
Let $t'_1, t'_2, \ldots$ be the time steps, 
such that at least one robot executes the Z-move in $P(t'_k)$ for ($k=1, 2, \ldots)$. 
We will show that $d_{t'_{k+1}} < d_{t'_k}$ for any $k = 1, 2, \ldots$, that is, 
the input distance to the Z-move monotonically decreases. 

We first bound the inter-robot distance in $t'_k +1$, 
i.e., after some Z-move in $t'_k$. 
We abuse the notation in the first paragraph: 
$p_i(t'_k)=p$, $p_j(t'_k)=q$, and 
$p'$ and $q'$ be the next positions of $r_i$'s Z-move 
and $r_j$'s next move, respectively. 
If $r_i$ ($r_j$, respectively) is not activated in $t'_k$, 
$p = p'$ ($q=q'$, respectively). 

In the following, 
we use the coordinate system of Algorithm~\ref{alg:HWPTF-2} 
and use $Z^*_0$ to describe it. 
The trajectory of the Z-move for $r_i$ is the line segment $pp'$, 
represented by the following equation of a line: 
\begin{eqnarray*}
&& y = \frac{\sqrt{(d_{t'_k}+d)(d_{t'_k}-d)(3d-d_{t'_k})(3d+d_{t'_k})}}{3(d^2_{t'_k}-d^2)} \left( x + \frac{d_{t'_k}}{2}\right) \\ 
&& -\frac{d_{t'_k}}{2} \leq x \leq \frac{d^2_{t'_k}-3d^2}{4d_{t'_k}}.
\end{eqnarray*}
The trajectory of the Z-move for $r_j$ is the line segment $qq'$, 
represented by the following equation of a line: 
\begin{eqnarray*}
&& y = \frac{\sqrt{(d_{t'_k}+d)(d_{t'_k}-d)(3d-d_{t'_k})(3d+d_{t'_k})}}{3(d^2_{t'_k}-d^2)} \left( x - \frac{d_{t'_k}}{2}\right) \\ 
&& - \frac{d^2_{t'_k}-3d^2}{4d_{t'_k}} \leq x \leq -\frac{d_{t'_k}}{2}.
\end{eqnarray*}
When $r_i$'s $x$ coordinate changes by $\varepsilon_i$ by a Z-move in $P(t'_k)$, 
$r_i$'s coordinates in $P(t'_k +1)$ is 
\begin{equation*}
\left( -\frac{d_{t'_k}}{2} + \varepsilon_i, \frac{\sqrt{(d_{t'_k}+d)(d_{t'_k}-d)(3d-d_{t'_k})(3d+d_{t'_k})}}{3(d_{t'_k}^2-d^2)} \varepsilon_i \right). 
\end{equation*}
When $r_j$'s $x$ coordinate changes by $\varepsilon_j$ by a Z-move in $P(t'_k)$, 
$r_j$'s coordinates in $P(t'_k +1)$ is 
\begin{equation*}
\left( \frac{d_{t'_k}}{2} - \varepsilon_j, \frac{\sqrt{(d_{t'_k}+d)(d_{t'_k}-d)(3d-d_{t'_k})(3d+d_{t'_k})}}{3(d_{t'_k}^2-d^2)} \varepsilon_j \right). 
\end{equation*}
Thus, we obtain 
\begin{equation*}
d_{t'_k +1} = \sqrt{\left( \varepsilon_i + \varepsilon_j - d_{t'_k} \right)^2 + \frac{(3d-d_{t'_k})(3d+d_{t'_k})}{9(d_{t'_k}^2 - d^2)} \left(\varepsilon_i + \varepsilon_j\right)^2}
\end{equation*}
Let $\varepsilon = \varepsilon_i + \varepsilon_j$. 
We also have the following inequality. 
\begin{equation*}
0 < \varepsilon \leq \frac{3d_{t'_k}^2-3d^2}{2d_{t'_k}}. 
\end{equation*}
Putting all together, we obtain 
\begin{equation*}
\sqrt{\frac{8 d^2_{t'_k}}{9(d_{t'_k}^2-d^2)} \left( \varepsilon - \frac{9(d_{t'_k}^2-d^2)}{8 d_{t'_k}}\right)^2 + \frac{9d^2 - d_{t'_k}^2}{8}}
\end{equation*}
Hence, we obtain 
\begin{equation*}
\sqrt{\frac{9d^2-d_{t'_k}^2}{8}} \leq d_{t'_k +1} < d_{t'_k}. 
\end{equation*}

We then show $d_{t'_{k+1}} < d_{t'_k}$, 
that is, the input distance to the $(k+1)$-st Z-move 
is smaller than that of the $k$-th Z-move. 

\noindent{Case 1.} When $d_{t'_{k+1}} > d$, 
at least one of the two robots performs the Z-move in $P(t'_k+1)$, i.e., $t'_{k+1} = t'_k +1$).
Hence, we have $d_{t'_{k+1}} < d_{t'_k}$. 

\noindent{Case 2.} When $d_{t'_{k+1}} < d$, 
at least one of the two robots performs the expanding move in $P(t'_k +1)$. 
By Lemma~\ref{lemma:expand}, there exists $\ell$ such that
\begin{equation*}
d \leq d_{t'_k+1 +\ell} < 3d 
\end{equation*}
and $t'_{k+1} = t'_k+1 +\ell$. 
By Algorithm~\ref{alg:HWPTF-2}, $P(t'_k +2)$ obtained by 
the expanding move in $P(t'_k +1)$ satisfies  
\begin{equation*}
d_{t'_k +1} < d_{t'k +2} \leq 2d - d_{t'k +1}. 
\end{equation*}
By repeating this equation, 
\begin{eqnarray*}
d_{t'_k+1} &<& d_{t'_k+2} \leq 2d - d_{t'_k+1} \\ 
d_{t'_k+2} &<& d_{t'_k+3} \leq 2d - d_{t'_k+2} \\ 
& \cdots & \\ 
d_{t'_k+1+ \ell} &<& d_{t'_k+\ell} \leq 2d - d_{t'_k+1 + \ell} 
\end{eqnarray*}
Hence, we have $d_{t'_k+1+ \ell} = d_{t'_{k+1}} < 2d - d_{t'_{k+1}}$. 
The maximum value of $2d - d_{t'_{k+1}}$ is achieved when $d_{t'_k +1}$ 
takes the minimum value  
\begin{equation*}
d_{t'_{k+1}} < 2d - d_{t'_{k+1}} \leq 2d - \sqrt{\frac{9d^2 - d_{t'_k}^2}{8}}. 
\end{equation*}
We will show 
\begin{equation}  
2d - \sqrt{\frac{9d^2 - d_{t'_k}^2}{8}} < d_{t'_k}.
\end{equation}
By changing the inequality, we obtain 
\begin{equation} \label{equ:diff}
2d - d_{t'_k} < \sqrt{\frac{9d^2 - d_{t'_k}^2}{8}}
\end{equation}
When $2d < d_{t'_k} < 3d$, the right-hand side of equation \eqref{equ:diff} 
is positive and the left-hand side is negative. 
Hence, equation \eqref{equ:diff} holds. 
When $d < d_{t'_k} \leq 2d$, by squaring both sides we obtain 
\begin{equation*}
9 \left( d_{t'_k} - \frac{16}{9} d \right) ^2 - \frac{49}{9} d^2  < 0 
\end{equation*}
This equation holds for $d < d_{t'_k} \leq 2d$. 
Hence, for $d < d_{t'_k} < 3d$, equation \eqref{equ:diff} holds. 
Consequently, we have $d_{t'_{k+1}} < d_{t'_k}$. 
That is, the input to the Z-move monotonically decreases 
even when the expanding moves are performed between 
two consecutive Z-moves. 

By the above discussion, we have 
$d_{t'_1} > d_{t'_2} > \ldots$. 
By equation~\ref{equ:zdist}, 
the moving distance of the Z-move monotonically decreases, 
and there exists finite $m$ such that in $P{t'_m}$ 
the moving distance of the Z-move is smaller than 
the minimum moving distance $\delta$. 
Then, the moving robots reaches its destination 
and $d_{t'_m + 1} = d$ holds. 
Of course, there might be other possibilities that 
the robots finish \HwPTF$(2,d)$ by the expanding move or 
the shrinking move. 
\qed
\end{proof}

By Lemma~\ref{lemma:expand}, \ref{lemma:shrink}, and \ref{lemma:zmove}, 
we have Theorem~\ref{thm:HWPTF-2d}.

\subsection{Two Oblivious ASYNC Robots}
\label{subsec:ASYNC-2}

In this section, we show that two oblivious ASYNC robots cannot solve \HwPTF$(2,d)$ 
by the fact that the two ASYNC robots must perform the Z-move in the final step. 
By definition, any SSYNC execution of some algorithm $A$ is also an ASYNC execution of $A$. 
By the calculation of the Z-move for any SSYNC (thus ASYNC) activation, there exists no other next position except that of the Z-move that solves \HwPTF$(2,d)$. 
On the other hand, Z-move does not consider all ASYNC schedules, that allows some robot to observe the position of moving robots. 
Each robot can observe the position of other robots, 
but cannot recognize which robot is moving. 
Hence, a robot may compute its next position for the Z-move without knowing whether the other robot is moving or not. 
We will first show that two robots cannot stop at distance $d$ under such ASYNC schedule even when the robots always perform rigid movement. 

\begin{lemma}
\label{lemma:Zmove-in-ASYNC}
Two oblivious deterministic ASYNC robots with rigid movement 
cannot finish \HwPTF$(2, d)$ by the Z-move for any $d$. 
\end{lemma}
\begin{proof}
The impossibility for $d=0$ is clear from Theorem~\ref{thm:HWPTF-2d} 
because the set of all executions of the ASYNC model 
contains that of the SSYNC model, i.e., 
an impossibility result in the SSYNC model directly applies to the ASYNC model. 

We consider the following scenario shown in Figure~\ref{fig:ASYNCtriangles}. 
\begin{enumerate}
\item In an initial configuration, $r_i$ is located at $p = (-a/2, 0)$ 
and $r_j$ is located at $q = (a/2, 0)$ ($d < a <3d$). 
\item First, $r_i$ is activated and computes its destination $p'$
of the Z-move, i.e., 
$$ p' = \left( \frac{a^2-3d^2}{4a}, \frac{\sqrt{(a^2-d^2)(9d^2-a^2)}}{4a} \right).$$
\item After $r_i$ has started its Z-move to $p'$, 
$r_j$ observes $r_i$ at position $p''$ such that $dist(p', p'') = \delta_a$. 
Then, $r_j$ computes its destination $q''$ of the Z-move.
\end{enumerate}
 \begin{figure}[t]
  \centering
  \includegraphics[width=4cm]{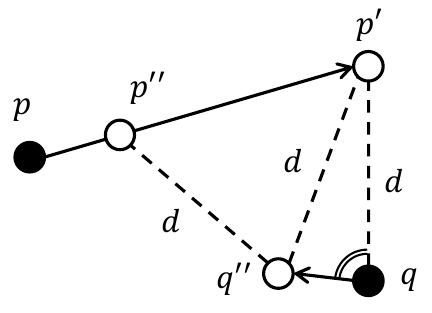} 
  \caption{Z-moves under the ASYNC schedule}
  \label{fig:ASYNCtriangles} 
\end{figure}

Assume that $r_i$ and $r_j$ succeeds in \HwPTF$(2, d)$ in the above scenario, 
that is, $dist(p', q'') = d$. 
We first calculate the cosine of $\theta = \angle{p' q q''}$ under the above condition. 
Because $p''$ lies on line $pp'$, the coordinate of $p''$ is
$$\left( \frac{a^2-3d^2}{4a}-\delta_a\sqrt{\frac{9(a^2-d^2)}{8a^2}}, \frac{\sqrt{(a^2-d^2)(9d^2-a^2)}}{4a}-\delta_a\sqrt{\frac{9d^2-a^2}{8a^2}} \right).$$
We can calculate $dist(q, p'')$ from the coordinates of $q$ and $p''$.
$$dist(q p'') = \sqrt{d^2+\delta_a^2+\delta_a\sqrt{\frac{a^2-d^2}{2}}}$$
For $\triangle{p'qp''}$, using the law of cosines, 
$$\cos{\angle{p' q p''}}=\frac{2d^2+\delta_a\sqrt{\frac{a^2-d^2}{2}}}{2d \left(\sqrt{d^2+\delta_a^2+\delta_a\sqrt{\frac{a^2-d^2}{2}}}\right)}.$$
We can calculate $\sin{\angle{p' q p''}}$ from $\cos{\angle{p' q p''}}$.
$$\sin{\angle{p' q p''}}=\frac{\delta_a\sqrt{\frac{9d^2-a^2}{2}}}{2d \left(\sqrt{d^2+\delta_a^2+\delta_a\sqrt{\frac{a^2-d^2}{2}}} \right)}$$

In an initial configuration where the inter-robot distance is $a$, 
the cosine of the angle formed by the Z-move and the $x$-axis of $Z_0^*$ is 
$\frac{3}{2a}\sqrt{\frac{a^2-d^2}{2}}$, 
and its sine is $\frac{1}{2a}\sqrt{\frac{9d^2-a^2}{2}}$.
By replacing $a$ with $dist(q p'')$, we obtain 
$$\cos{\angle{q''qp''}}=\frac{3}{2\sqrt{2}}\sqrt{\frac{\delta_a^2+\delta_a\sqrt{\frac{a^2-d^2}{2}}}{d^2+\delta_a^2+\delta_a\sqrt{\frac{a^2-d^2}{2}}}}$$
$$\sin{\angle{q''qp''}}=\frac{1}{2\sqrt{2}}\sqrt{\frac{8d^2-\delta_a^2-\delta_a\sqrt{\frac{a^2-d^2}{2}}}{d^2+\delta_a^2+\delta_a\sqrt{\frac{a^2-d^2}{2}}}}.$$
Using the addition formulas of cosines,
\begin{equation}
\label{eq:cos1}
\cos{\angle{p' q q''}}=\frac{3 \left( 2d^2+\delta_a\sqrt{\frac{a^2-d^2}{2}} \right)\sqrt{\delta_a^2+\delta_a\sqrt{\frac{a^2-d^2}{2}}}-\delta_a\sqrt{\frac{9d^2-a^2}{2}}\sqrt{8d^2-\delta_a^2-\delta_a\sqrt{\frac{a^2-d^2}{2}}}}{4\sqrt{2}\left( d^2+\delta_a^2+\delta_a\sqrt{\frac{a^2-d^2}{2}}\right)}.
\end{equation}

We than calculate the cosine of $\angle{p' q q''}$ under the condition 
that $dist(p' q'') = d$. 
We also have $dist(p' q) = d$ by the Z-move and 
$$dist(qq'') = \sqrt{\frac{dist(qp'')^2-d^2}{2}}=\sqrt{\frac{\delta_a^2+\delta_a\sqrt{\frac{a^2-d^2}{2}}}{2}}.$$
$\triangle{p'qq''}$ is an isosceles triangle.
So, we can calculate the cosine of $\angle{p' q q''}$ directly.
\begin{equation}
\label{eq:cos2}
\cos{\angle{p' q q''}}=\frac{\sqrt{\delta_a^2+\delta_a\sqrt{\frac{a^2-d^2}{2}}}}{2\sqrt{2}d}. 
\end{equation}

For simplicity, let $a = sd$ ($1<s<3$) and 
$\delta_a = t \sqrt{(a^2-d^2)/2} = td \sqrt{(s^2-1)/2}$ ($0<t<1$).
We obtain the following equation by equations \eqref{eq:cos1} and \eqref{eq:cos2}. 
\begin{eqnarray*}
&&\frac{3(4+t(s^2-1))\sqrt{(t^2+t)(s^2-1)}-t\sqrt{(s^2-1)(9-s^2)(16-(t^2-t)(s^2-1))}}{16+8(t^2+t)(s^2-1)} \\
&=&\frac{\sqrt{(t^2+t)(s^2-1)}}{4}
\end{eqnarray*}
When $t = (s-1)/2$ ($1 < s< 3$), the equation does not hold. 
Consequently, the ASYNC scheduler can forever prevent the two robots 
from solving \HwPTF$(2, d)$ for and $d>0$. 
\qed
\end{proof}

\begin{theorem}
\label{thm:impossibleObliviousASYNC}
Two oblivious deterministic ASYNC robots cannot solve \HwPTF$(2, d)$ for any $d > \delta$. 
\end{theorem}
\begin{proof}
Assume that there exists an algorithm $A$ that solves \HwPTF$(2,d)$ 
from an arbitrary initial configuration for any $d$ larger than 
the minimum moving distance $\delta$. 
We will show that we can construct an infinite execution of $A$ 
where the two robots never finish \HwPTF$(2,d)$. 
Consider an execution $P(0), P(1), P(2), \ldots$ 
starting from an initial configuration $P(0)$, 
and the sequence of inter-robot distance $d_0, d_1, d_t \ldots$ 
at time $t$ ($t = 0, 1, 2, \ldots $ ). 
We have the following three cases depending on the value of $d_t$. 
\begin{itemize}
\item \textbf{Case 1. $3d \leq d_t$.} 
If the two robots finish \HwPTF$(2,d)$ in $P(t+1)$, 
at least one robot moves and its moving distance is equal to or larger than $d$. 
By $d>\delta$, the adversary can stop that robot before it reaches its destination. 
Hence, we have an execution where $d_{t+1} \neq d$. 
\item \textbf{Case 2. $d < d_t <3d$.}  
If the two robots finish \HwPTF$(2,d)$ after $t$, 
as already discussed, the two robots $r_i$ and $r_j$ must perform the Z-move. 
By Lemma~\ref{lemma:Zmove-in-ASYNC}, there exists an activation schedule, 
where $r_i$ first observes $r_j$ and starts a Z-move, 
and then $r_j$ observes $r_i$ en route and starts another Z-move, 
so that when $r_i$ and $r_j$ reach their destinations, 
the inter-robot distance is not $d$. 
\item \textbf{Case 3. $d_t < d$.} 
If the two robots finish \HwPTF$(2,d)$ in $P(t+1)$, 
at least one robot moves away from the other robot. 
However, if $r_i$ moves away from $r_j$ when the inter-robot distance is $d_t$, 
$r_j$ may perform the same movement 
when $r_j$ is also activated at $t$ and obtain the same observation as $r_i$. 
Hence, if \HwPTF$(2,d)$ is finished in $P(t+1)$ by the movement of either $r_i$ or $r_j$, 
there exists another execution where the adversary activates both robots 
in $P(t)$, thus $d_{t+1} \neq d$. 
Otherwise, \HwPTF$(2,d)$ is finished in $P(t+1)$ by the simultaneous and same 
movement of $r_i$ and $r_j$. 
There exists another execution where the adversary activates only one robot 
in $P(t)$, thus $d_{t+1} \neq d$. 
\end{itemize}
Hence, we can construct an infinite execution of $A$, 
where $r_i$ and $r_j$ never finish \HwPTF$(2,d)$. 
\qed
\end{proof}

\subsection{Two Luminous ASYNC Robots}

In this section, we consider robots that are endowed with a \emph{light} whose color belongs to a fixed palette~\cite{DFPSY16}.
Robots can see the light color of other robots, and modify the color of their own light during the compute phase.
We present a positive result that two luminous deterministic robots can solve the \HwPTF \ problem, provided they share the same unit distance. 
Our approach is based on a bijection between a class of algorithms solving the rendezvous problem and their transformed version solving the \HwPTF \ for two robots. 

\subsubsection{Luminous ASYNC Rendezvous Algorithms}

Let $D > 0$ denote the current inter-robot distance, with robot~$r_i$ at position
$0$ and robot~$r_j$ at position~$D$ in~$r_i$'s local coordinate system.
The rendezvous algorithms we consider are those that use only three kinds of moves:
\begin{description}
  \item[{\STAY}] The robot does not move.
  \item[\HALF\; (move to midpoint)] The robot moves toward the other robot to the
        midpoint: target position~$D/2$. If the robot is not interrupted, the new inter-robot distance is~$D/2$.
  \item[\OTHER\; (move to other)] The robot moves to the other robot's current position: target position~$D$. If the robot is not interrupted, the new inter-robot distance is~$0$.
        
\end{description}
These three movements are necessary and sufficient to achieve
gathering~\cite{Her2Cols}.

A luminous \emph{rendezvous algorithm} $\mathcal{A}$ for the ASYNC model is expressed
as a table of \emph{guarded commands}: each row maps a guard (typically the
pair (my color, other's color)) to a new color and a movement chosen from
\STAY, \HALF, \OTHER.  The guard ``gathered'' holds when both robots occupy
the same position; ``skip'' means no change; ``--'' retains the current color;
rule precedence is top-to-bottom.  
As a running example we use the
\emph{Her2Cols} algorithm of Heriban et al.~\cite{Her2Cols}:

\begin{center}
\renewcommand{\arraystretch}{1.2}
\begin{tabular}{l@{\quad$\leadsto$\quad}l}
  \hline
  \textbf{Guard} & \textbf{New color,\ movement} \\
  \hline
  $(\BLACK,\BLACK)$ & \WHITE, \STAY \\
  $(\BLACK,\WHITE)$ & skip \\
  gathered          & skip \\
  $(\WHITE,\BLACK)$ & --, \OTHER \\
  $(\WHITE,\WHITE)$ & \BLACK, \HALF \\
  \hline
\end{tabular}
\end{center}

The Her2Cols algorithm solves \emph{deterministic rendezvous in ASYNC}~\cite{Her2Cols} in a self-stabilizing manner (that is, starting from
any initial configuration~\cite{Dijkstra74}, regardless of the ASYNC scheduler).

\subsubsection{From Rendezvous ASYNC Algorithms to ASYNC \HwPTF$(2, d)$}

We first present new primitive movements that are dedicated to the \HwPTF$(2, d)$. 
Let $D > 0$, $D \neq d$, denote the current inter-robot distance.
Define the \emph{signed excess distance}~$D' = D - d$.  We have $D' \in
\mathbb{R} \setminus \{0\}$; in particular $D' < 0$ when $D < d$.

In $r_i$'s local coordinate system (with $r_i$ at~$0$ and $r_j$ at~$D$), the two adapted movements are:

\begin{description}
  \item[\HALFD\; (halve the signed excess)]
    Robot~$r_i$ moves to the target position~$(D-d)/2 = D'/2$.
    \begin{itemize}
      \item If $D > d$: this point is strictly between $r_i$ and $r_j$; the robot
            moves \emph{toward} the other robot.
      \item If $D < d$: this point lies on the opposite side of $r_i$ from $r_j$
            (i.e., $(D-d)/2 < 0$); the robot moves \emph{away} from the other
            robot.
    \end{itemize}
    In both cases the new inter-robot distance is
    $\bigl|D - (D-d)/2\bigr| = (D+d)/2$,
    and the new signed excess is $D'/2$.

  \item[\OTHERD\; (reach distance exactly~$d$)]
    Robot~$A$ moves to the target position~$D - d$.
    \begin{itemize}
      \item If $D > d$: this point lies strictly between $r_i$ and $r_j$; the robot
            moves \emph{toward} the other robot.
      \item If $D < d$: this point lies on the opposite side of $r_i$ from~$r_j$
            (i.e., $D - d < 0$); the robot moves \emph{away} from the other
            robot.
    \end{itemize}
    In both cases the new inter-robot distance is $|D - (D-d)| = d$,
    and the new signed excess is~$0$.
\end{description}

Both targets are computable and consistent for the two robots  from the local snapshot alone: robot~$r_i$ observes
the distance~$D$ to~$r_j$, and~$d$ is a fixed parameter of the hardwired pattern. 
Since both
robots share the same unit distance, the value of~$d$ is identical in the
local coordinate systems of both robots.

Setting $D' = D - d$, the movements \HALFD \ and \OTHERD \ act on the signed excess
exactly as \HALF \ and \OTHER \ act on the inter-robot distance in the rendezvous
setting:
\[
  \HALFD:\; D' \longmapsto \frac{D'}{2}, \qquad
  \OTHERD:\; D' \longmapsto 0.
\]
The goal condition ``at distance~$d$'' corresponds to $D' = 0$, exactly as
``gathered'' corresponds to $D = 0$.

Given a gathering algorithm $\mathcal{A}$, define its \emph{adapted
algorithm}~$\mathcal{A}_d$ by applying the following three syntactic
substitutions to every rule of~$\mathcal{A}$:
\begin{enumerate}
  \item Replace every guard ``gathered'' with ``at distance~$d$''.
  \item Replace every movement \HALF\; with \HALFD.
  \item Replace every movement \OTHER\; with \OTHERD.
\end{enumerate}
The color palette, color guards, and color transitions are left unchanged.

Applying this substitution to Her2Cols yields Her2Cols-$d$.
The two algorithms are shown side by side below; only the three highlighted
entries differ.

\medskip
\begin{center}
\renewcommand{\arraystretch}{1.2}
\begin{tabular}{l@{\quad$\leadsto$\quad}l@{\qquad\qquad}l@{\quad$\leadsto$\quad}l}
  \hline
  \multicolumn{2}{c}{\textbf{Her2Cols} (rendezvous)}
  & \multicolumn{2}{c}{\textbf{Her2Cols-$d$} (exact distance~$d$)} \\
  \hline
  \textbf{Guard} & \textbf{Color, move}
  & \textbf{Guard} & \textbf{Color, move} \\
  \hline
  $(\BLACK,\BLACK)$ & \WHITE, \STAY
  & $(\BLACK,\BLACK)$  & \WHITE, \STAY \\
  $(\BLACK,\WHITE)$ & skip
  & $(\BLACK,\WHITE)$  & skip \\
  $\boldsymbol{\mathit{gathered}}$           & skip
  & $\boldsymbol{\mathit{at\ distance}\ d}$    & skip \\
  $(\WHITE,\BLACK)$ & --, \textbf{\OTHER}
  & $(\WHITE,\BLACK)$  & --, $\textbf{\OTHERD}$ \\
  $(\WHITE,\WHITE)$ & \BLACK, \textbf{\HALF}
  & $(\WHITE,\WHITE)$  & \BLACK, $\textbf{\HALFD}$ \\
  \hline
\end{tabular}
\end{center}
\medskip

\noindent
Guards match (my color, other's color); ``skip'' means no change;
``--'' retains the current color; rule precedence is top-to-bottom.

We now state our positive result.

\begin{theorem}[ASYNC Rendezvous $\Rightarrow$ ASYNC \HwPTF$(2, d)$]
\label{thm:transfo}
Let $d > 0$.  If a deterministic algorithm~$\mathcal{A}$ using movements
\STAY, \HALF, \OTHER\;  solves rendezvous in the ASYNC model (starting from any
initial configuration allowed by~$\mathcal{A}$), then its adapted
algorithm~$\mathcal{A}_d$ solves the \HwPTF$(2, d)$ in the ASYNC
model, starting from any initial configuration with $D > 0$ that corresponds to an allowed initial configuration of~$\mathcal{A}$ via the
map $D \mapsto D' = D - d$.
\end{theorem}

\begin{proof}
We exhibit a bijection between executions of~$\mathcal{A}_d$ and executions
of~$\mathcal{A}$ that preserves the scheduler, the color sequences, and the
goal-achievement event.

\medskip
\noindent\textbf{Step 1: Isomorphism of individual movements.}
Fix any execution of~$\mathcal{A}_d$ with initial distance $D_0 > 0$,
$D_0 \neq d$.  Define the \emph{signed excess distance at time $t$} $D'_t = D_t - d$, where
$D_t$ is the inter-robot distance at time~$t$.

\begin{itemize}
  \item \HALFD \; maps $D'_t \mapsto D'_t/2$, identical to the action of \HALF \; on
        the inter-robot distance in the rendezvous setting.
  \item \OTHERD\; maps $D'_t \mapsto 0$, identical to the action of \OTHER.
  \item \STAY\; leaves $D'_t$ unchanged, identical to \STAY.
\end{itemize}
These identities hold regardless of the sign of $D'_t$, i.e., for both
$D > d$ and $D < d$.

\medskip
\noindent\textbf{Step 2: Isomorphism of movement resolution.}
The movement resolution rules of considered rendezvous algorithms govern what happens when two robots execute movements concurrently.  The critical case is when both robots execute \HALF\; simultaneously: the companion robot's pending move is upgraded to
\OTHER\; so that both robots eventually reach the same point.

The analogous case in~$\mathcal{A}_d$ is when both robots execute \HALFD\;
simultaneously.  In $r_i$'s coordinate system, $r_i$ starts at $0$ and $r_j$ at $D$.
$r_i$ moves to $(D-d)/2$.  The companion robot $r_j$'s pending move is upgraded
to \OTHERD: $r_j$ moves to the point at distance~$d$ from $r_i$'s \emph{new}
position $(D-d)/2$, in the direction of $r_j$ from $r_i$, i.e., to
$(D-d)/2 + d = (D+d)/2$.
$r_j$ therefore moves from $D$ to $(D+d)/2$, which corresponds to moving by
$(D-d)/2$ \emph{toward} $r_i$ when $D > d$ (since $(D+d)/2 < D$) and
\emph{away} from $r_i$ when $D < d$ (since $(D+d)/2 > D$).
The new inter-robot distance is $|(D+d)/2 - (D-d)/2| = d$. 

In general, the movement resolution rules for~$\mathcal{A}_d$ are obtained from
those of~$\mathcal{A}$ by replacing every occurrence of \HALF\; by \HALFD\;, every
occurrence of \OTHER\; by \OTHERD\;, and every occurrence of ``gathered''
by ``at distance~$d$'' (\ATDIST), and interpreting movements in terms of $D'$.
Because the excess distance at time $t$ $D'_t$ follows the same arithmetic as $D_t$ in the rendezvous setting, the outcome of every movement resolution event in~$\mathcal{A}_d$
is identical (in terms of~$D'$) to the outcome of the corresponding event
in~$\mathcal{A}$.

\medskip
\noindent\textbf{Step 3: Isomorphism of color dynamics.}
By Definition, the color guards and color transitions
of~$\mathcal{A}_d$ are identical to those of~$\mathcal{A}$.  The only
additional guard, ``at distance~$d$'', corresponds (via $D' = D - d$) to
$D' = 0$, which corresponds to ``gathered'' in~$\mathcal{A}$.  Hence the color
sequence produced by~$\mathcal{A}_d$ on any execution is identical to the color
sequence produced by~$\mathcal{A}$ on the corresponding execution in terms
of~$D'$.

\medskip
\noindent\textbf{Step 4: Conclusion.}
By Steps 1--3, there is a bijection $\varphi$ between executions of~$\mathcal{A}_d$
starting with excess $D'_0 = D_0 - d$ and executions of~$\mathcal{A}$ starting
with distance $D_0$, that preserves the scheduler and the color sequence.
Under~$\varphi$, the event ``$\mathcal{A}_d$ reaches the goal $D = d$ (i.e.,
$D' = 0$)'' corresponds exactly to the event ``$\mathcal{A}$ achieves rendezvous
($D = 0$)''.

Since $\mathcal{A}$ solves rendezvous from any allowed initial configuration,
$\mathcal{A}$ drives~$D_0$ to~$0$ in every fair execution.
Equivalently, under~$\varphi$, $\mathcal{A}_d$ drives $D'_0 = D_0 - d$ to~$0$
in every fair execution, i.e., it drives $D_0$ to~$d$.
Therefore, $\mathcal{A}_d$ solves \HwPTF$(2, d)$ in~ASYNC. 
\qed
\end{proof}

Applying Theorem~\ref{thm:transfo} to the Her2Cols algorithm~\cite{Her2Cols}, and observing that oblivious robots (so, with a single color) cannot solve \HwPTF \ for two robots (Theorem~\ref{thm:impossibleObliviousASYNC}), yields the following corollary:
\begin{corollary}[Her2Cols-$d$]
\label{cor:her}
Her2Cols-$d$ solves the \HwPTF \ problem for two robots in the
ASYNC model, for any initial inter-robot distance $D > 0$.
Moreover, it uses only two colors (which is optimal) and is self-stabilizing.
\end{corollary}

Observe that our transformation applies to \emph{every} deterministic
gathering algorithm using only \STAY, \HALF, \OTHER—a necessary condition for
gathering~\cite{Her2Cols}.  This covers, in particular, Vig3Cols (for ASYNC)~\cite{Viglietta13}, and
Her2Cols (for ASYNC)~\cite{Her2Cols}: in each case the adapted algorithm
solves exact-distance-$d$ under the same synchrony model with the same number
of colors.

\section{Hardwired Pattern Formation for More than Four SSYNC Robots}
\label{sec:adjust}

In this section, we present a \HwPTF \ algorithm 
for oblivious SSYNC robots. 
We present a ``size-adjusting'' algorithm 
that shrinks or expands the SEC of the robots to a specified size $d$. 
We can obtain a \HwPTF \ algorithm 
by combining the proposed algorithm and 
some existing pattern formation algorithm~\cite{SY99}, 
that does not change the SEC of the robots, i.e., 
the robots execute the proposed algorithm until their size becomes $d$, 
and a pattern formation algorithm after that. 

Existing papers showed that the robots can form a target pattern $F$
from an initial configuration $I$ only if $\rho(F)$ is divisible by $\rho(I)$ 
due to impossibility of symmetry breaking~\cite{FYOKY15,SY99,YS10}. 
That is, the robots at symmetric positions w.r.t. $\rho(I)$ 
forever perform symmetric movement in the worst case. 
This impossibility also holds for the \HwPTF \ problem. 
We will prove that the condition is also sufficient by a 
\HwPTF \ algorithm.

\subsection{Size-adjusting Algorithm}
\label{subsec:adjust}

Consider a naive shrinking algorithm that moves the robots 
toward the center of the SEC so that the size of a resulting configuration is $d$. 
This algorithm guarantees the convergence to size $d$, 
but not the adjustment to size $d$ due to SSYNC activation of robots. 
Another issue is how to prevent the symmetricity of the robots from increasing. 

Our key idea is to keep the center $c(P(0))$ of an initial configuration $P(0)$ 
by $\rho(P(0))$ robots closest to the center.  
Intuitively, the robots can adjust their size by moving 
to a point at distance $d/2$ from $c(P(0))$ along the radius of the SEC 
while the initial symmetricity is maintained by the $\rho(P(0))$ robots 
closest to the center. 

We explain the detail with the case where the robots shrink their size. 
There are two principles for the proposed algorithm. 
First, each robot moves on the radius, 
i.e., the half line starting from $c(P(0))$ and passing through its current position. 
Second, each radius has designated positions, 
i.e., positions at distance $d \sum_{j=1}^i (1/2)^{j+1}$ $(i=1, 2, \ldots)$ from $c(P(0))$. 
Let $h_i$ denote the $i$-th designated position. 
On each radius, the proposed algorithm relocates the $i$-th robot from $c(P(0))$ at $h_i$. 
However, the first and the last designated positions can be empty 
because some robots closest to $c(P(0))$ move closer to keep $c(P(0))$ 
and some robots stop at distance $d/2$ from $c(P(0))$ to finish the adjustment. 

The center of $P(0)$ can be easily maintained 
if the robots create a multiplicity at $c(P(0))$ and no other multiplicity is created.  
However, it is impossible to resolve a multiplicity by a deterministic algorithm 
once it is created. 
Instead, the proposed algorithm keeps $c(P(0))$ by the minimum distance 
between the robots. 
For a set $P$ of points, let $d_h(P,q,d)$ denote the minimum distance 
between the points in $P$ when all points are relocated at designated positions 
when the center (i.e., the starting point of radius) is $q$. 
That is, for each point $p \in P$, we consider the half line $qp$. 
If $p$ is the $i$-th point from $q$ on $qp$, its designated position 
is $h_i = d \sum_{j=1}^i (1/2)^{j+1}$ from $q$. 
Then, $d_h(P,q,d)$ is the minimum distance among the robots 
in such a virtual relocated configuration. 

Let $\{P(0)_1, P(0)_2, \ldots, P(0)_{n/\rho(P)}\}$ be the 
$\rho(P(0))$-decomposition of an initial configuration $P(0)$. 
The proposed algorithm first sends the robots in $P(0)_1$ toward $c(P(0))$ 
to create the minimum distance among the robots during any execution. 
The diameter of such regular $\rho(P(0))$-gon must be smaller than 
$d_{min}(P(0))$, $d/16$, and $d_h(P(0),c(P(0)),d)$. 
Hence, we adopt $d^* = (1/2)(\min\{d_{min}(P(0)), d/16, d_h(P(0),c(P(0)),d)\})$. 
When $\rho(P(0))=1$, the proposed algorithm selects the first two 
elements of the $\rho(P(0))$-decomposition of $P(0)$, that does not 
change $SEC(P(0))$ when they move to the interior of $SEC(P(0))$. 
Let $P(0)_{i}, P(0)_j$ ($i<j$) be these two elements. 
Then, $P(0)_i$ moves toward $c(P(0))$ along its radius 
but stops $\delta$ before the center. 
Then, $P(0)_j$ moves toward $c(P(0))$ along its radius 
and stops at the point where the inter-robot distance with $P(0)_i$ 
becomes smaller than $d^*$. 
Thus, the center is perturbed but it is not a problem as explained later. 

In an arbitrary configuration $P$, 
the robots can agree on the center encoded by the above method by the following 
procedure. 
First, each robot checks the pairs of robots achieving the minimum 
inter-robot distance. Let $k$ be the number of such robots. 
Second, it checks whether $k \leq n/2$. 
Third, it checks whether the $k$ robots form a regular $k$-gon 
containing no point of $P$. 
Finally, it checks whether the diameter of the SEC of the regular $k$-gon is smaller than $d/16$. 
If the regular $k$-gon satisfies all these conditions, 
the proposed algorithm considers its center as the center for adjustment, 
and we say the robots can \emph{recognize the center} in $P$.

We finally check the ``fictitious'' center in some initial configurations. 
For a set $P$ of points and a point $q \not\in P$, 
we consider the symmetricity and decomposition w.r.t. $q$. 
Consider a decomposition of $P$ into regular $k$-gons centered at $q$. 
The maximum value of such $k$ is the symmetricity of $P$ w.r.t. $q$ 
denoted by $\rho^*(P,q)$. 
We also define SEC and LEC w.r.t. $q$. 
The SEC of $P$ w.r.t. $q$ is the smallest circle centered at $q$ 
and contains all points of $P$ in its interior of boundary. 
The LEC of $P$ w.r.t. $q$ is the largest circle centered at $q$ 
and contains no point of $P$ in its interior. 
We have the following property. 
\begin{observation}
\label{obs:ficticious-center}
For any set $P$ of multiple points, 
$\rho^*(P,q) = 1$ when $q \neq c(P)$. 
\end{observation}
\begin{proof}
Let $st$ be the diameter of $SEC(P)$ perpendicular to line segment $q c(P)$. 
The arc $\stackrel{\frown}{st}$ of $SEC(P)$ (including the two endpoints) 
in the opposite side w.r.t. $q$ contains at least one point $u \in P$. 
If $\rho^*(P,q) > 1$, the circle $C'$ centered at $q$ with radius $dist(qu)$ 
contains at least one symmetric point $u'$ for $u$. 
The two circles $SEC(P)$ and $C'$ has at most two intersections, 
and let $u'$ be the another intersection. 
Because $\angle uqu' < \pi$, circle $C'$ must contain at least another point 
to have $\rho^*(P,q) > 1$, a contradiction. 
\qed
\end{proof}

Such a fictitious center can be generated by the movement toward $c(P)$ 
due to SSYNC activation. 
Whenever the center $q$ is recognized in a configuration $P$, 
the size-adjusting algorithm checks $\rho^*(P,q)$ 
and rearranges the robots that keeps the minimum distance around $q$ 
according to $\rho^*(P,q)$. 
Hence, symmetricity of the configuration does not increase during 
any execution. 

When a robot moves along its radius, 
it moves cautiously so that it does not create a new minimum distance. 
That is, on each radius if the $i$-th robot $r_i$ wants to move to its 
designated position $h_i$ 
and some other robots are on the trajectory, 
$r_i$ just wait these robots to reach their designated positions. 
This does not result in a deadlock because the robots do not need to 
pass the other robots. 

When the robots expand to size $d$, 
the center is recognized by a regular polygon of size 
smaller than $(1/2)(\min\{d_{min}(P), rad(P)/16\})$ 
and the robots are not required to relocate on designated positions. 

Consequently, the adjusting algorithm consists of the following 
three steps. 
Let $P$ be the current configuration and $P_1, P_2, \ldots, P_{n/\rho(P)}$ be 
the $\rho(P)$-decomposition of $P$. 
\begin{itemize} 
\item Step 1. Form a single center. 
When the robots cannot recognize a center in $P$, 
$P_1$ moves toward $c(P)$ to create a regular $\rho(P)$-gon 
with a size discussed above. 
\item Step 2. Move to designated positions when $rad(P) > d$. 
When the robots can recognize the single center $q$ in $P$, 
let $\{P_1', P_2', \ldots, P_{n/\rho^*(P,q)} \}$ be the decomposition of $P$ 
into these regular $\rho^*(P,q)$-gons. 
Then, each robot in $P \setminus P_1' \cup P'_{n/\rho^*(P,q)-1}$ move to their 
designated positions. 
\item Step 3. Shrink or expand w.r.t. the recognized center. 
When the robots can recognize the single center $q$ in $P$, 
and $P_2', P_3', \ldots, P_{n/\rho^*(P,q)-1}'$ are on their designated positions, 
the robots in $P_{n/\rho^*(P,q)}'$ move toward or away from $q$ 
to form a circle with diameter $d$ centered at $q$. 
\end{itemize}

Figure~\ref{fig:adjust} shows an example of shrinking execution of the adjusting algorithm. 
The red circle is a circle of diameter $d$, 
and the blue circle is a circle of diameter $d/2$. 
In an initial configuration $P(0)$ (Figure~\ref{fig:adjust} ($a$)), 
$\rho(P(0)) = 4$ and only two robots on $LEC(P)$ moves toward $c(P(0))$ in Step 1. 
The robots can recognize the center in $P(1)$ (Figure~\ref{fig:adjust} ($b$)). 
Then, other robots move to their designated positions 
(Figure~\ref{fig:adjust} ($a$)).
In $P(3)$ (Figure~\ref{fig:adjust} ($d$)), 
the robots finish Step 2 and 
the robots on $SEC(P(3))$ moves toward the circle of diameter $d$. 
\begin{figure}[t]
 \begin{tabular}{ccc}
 \begin{minipage}[t]{0.3\hsize}
   \centering
   \includegraphics[width=4cm]{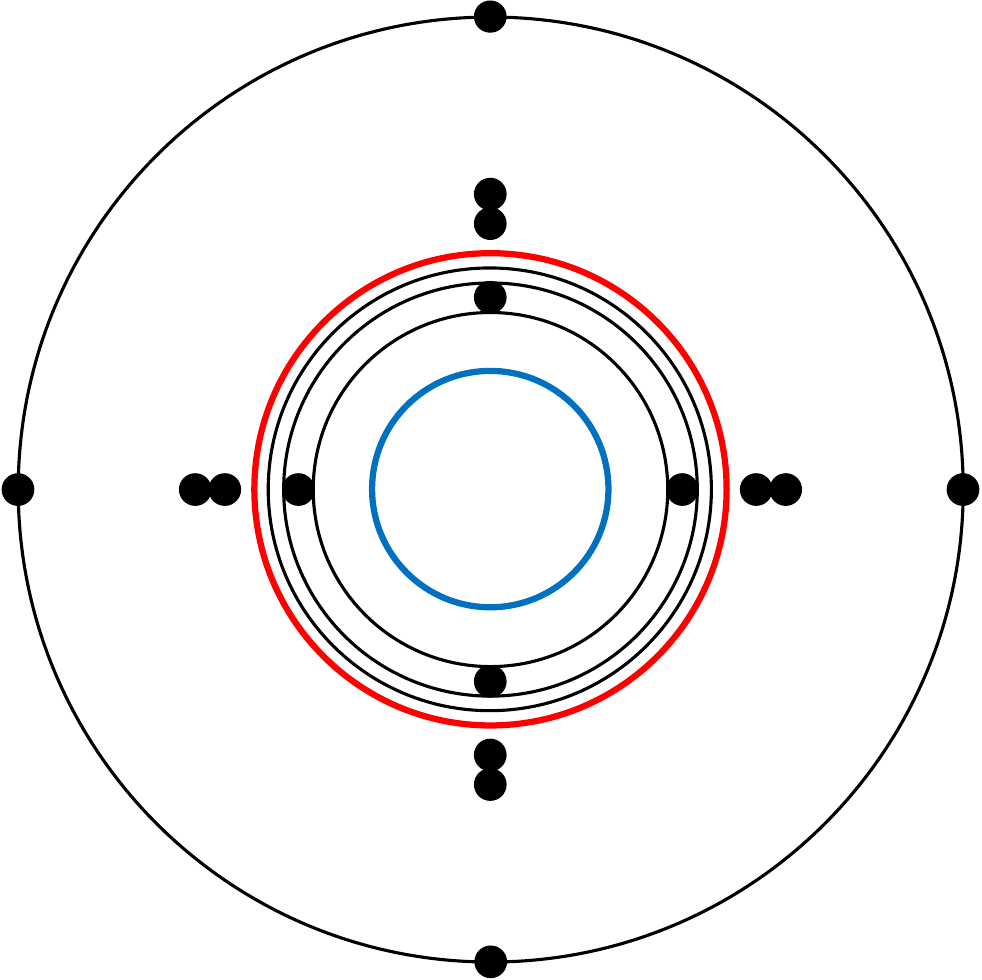} \\ 
   ($a$) $P(0)$
  \end{minipage} & 
  \begin{minipage}[t]{0.3\hsize}
   \centering
   \includegraphics[width=4cm]{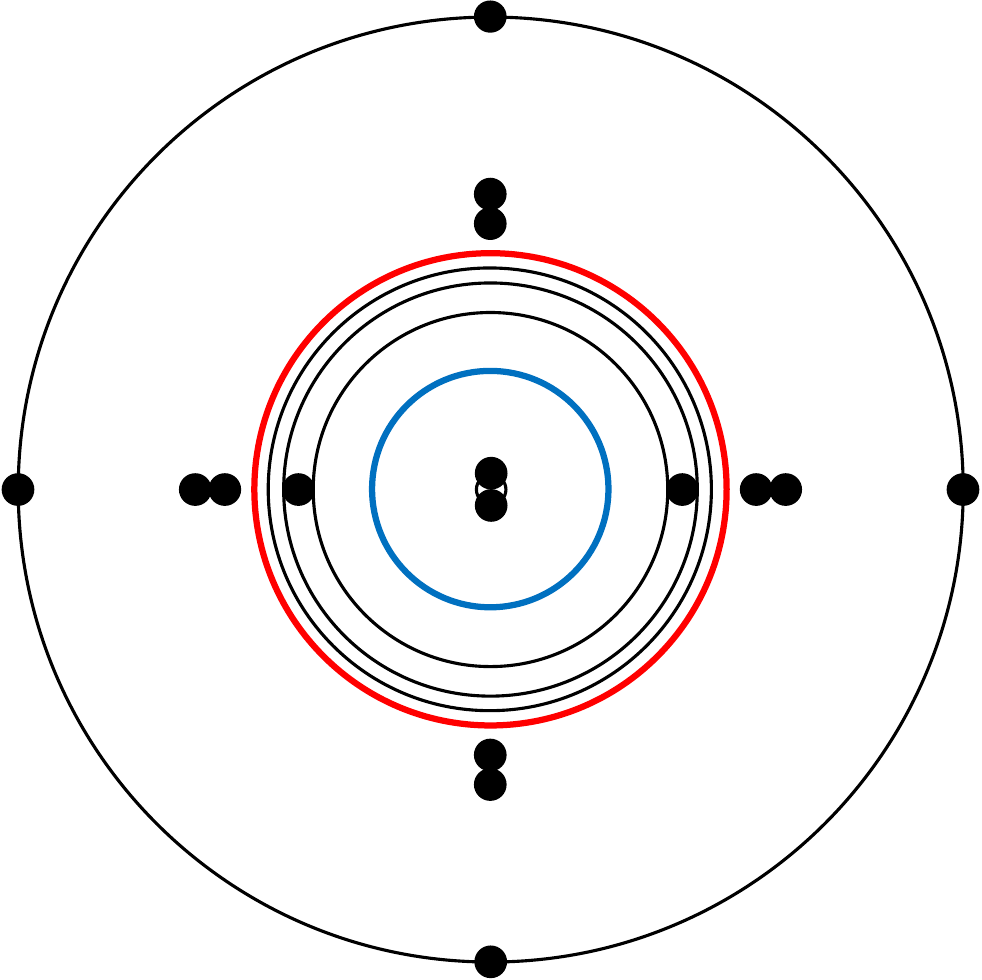} \\
   ($b$) $P(1)$
  \end{minipage} & 
  \begin{minipage}[t]{0.3\hsize}
   \centering
   \includegraphics[width=4cm]{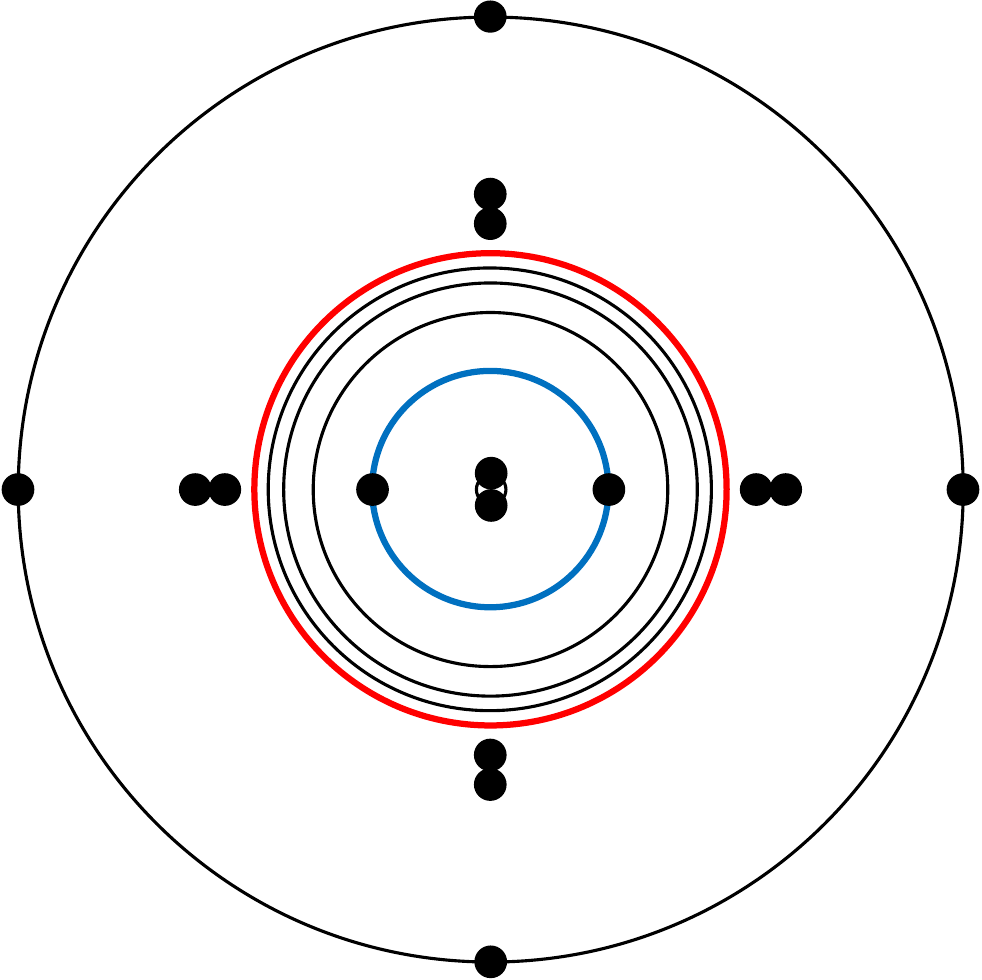} \\
   ($c$) $P(2)$
  \end{minipage} \\ 
  \begin{minipage}[t]{0.3\hsize}
   \centering
   \includegraphics[width=4cm]{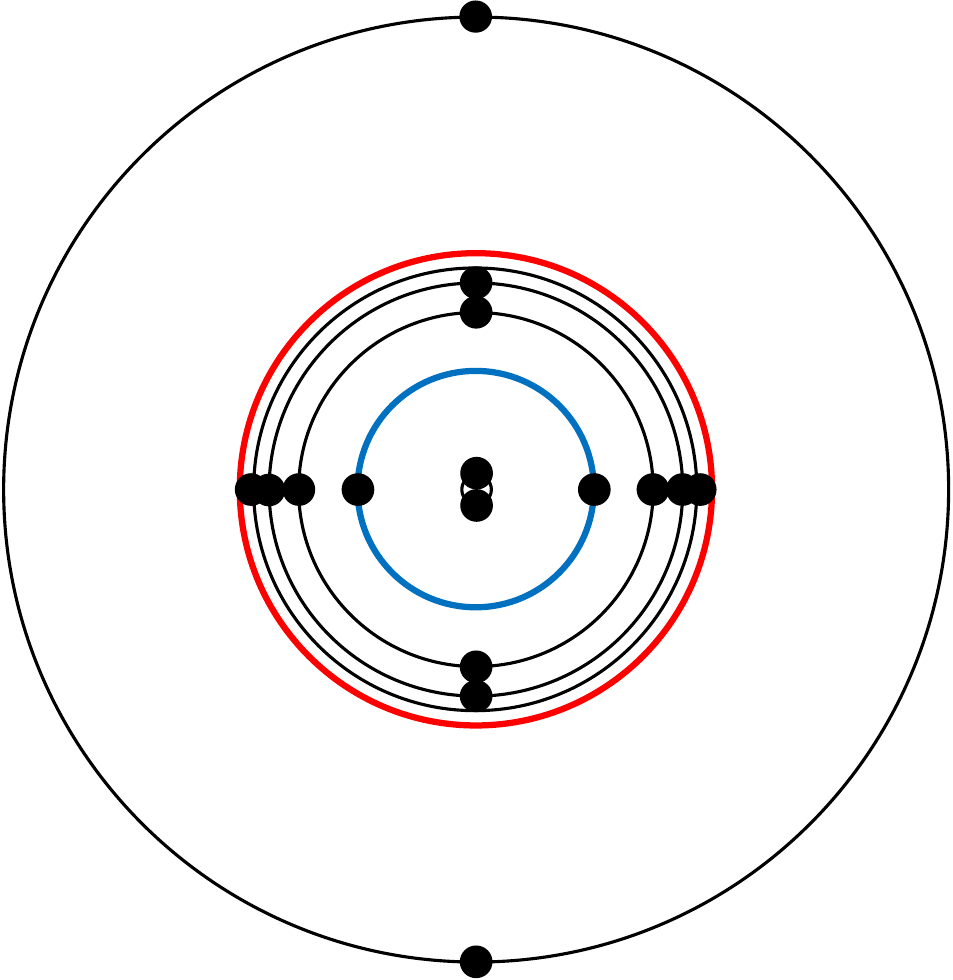} \\
    ($d$) $P(3)$
 \end{minipage} & 
  \begin{minipage}[t]{0.3\hsize}
   \centering
   \includegraphics[width=4cm]{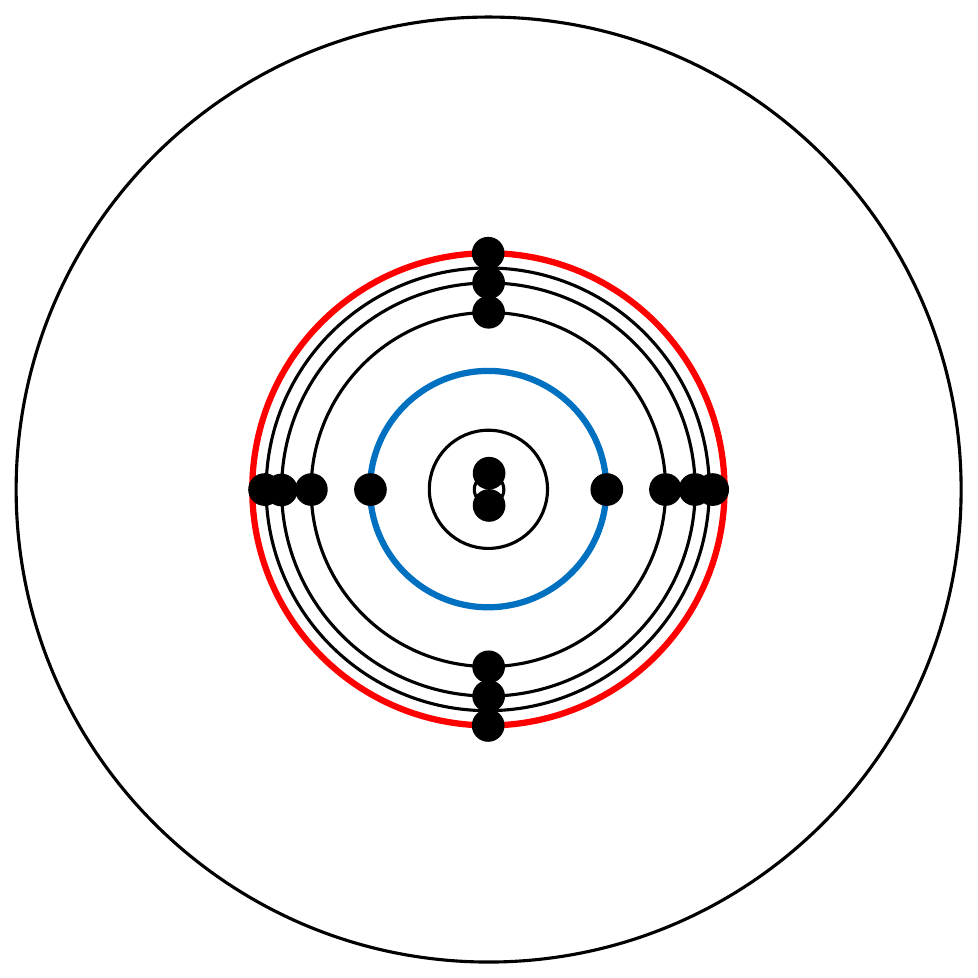}\\
      ($e$) $P(4)$
  \end{minipage} 
\end{tabular}
\caption{Example of the adjusting algorithm when the robots shrink}
\label{fig:adjust}
\end{figure}

The proposed algorithm works for more than four robots 
because when the symmetricity of the configuration $P$ is one, 
it sends two robots to the interior of SEC, 
which requires at most three robots to be kept. 

Algorithm~\ref{alg:SAdjust} shows the proposed size-adjusting algorithm 
for robot $r_i$. 

\begin{algorithm}
\caption{\textsc{Hardwired\_Pattern\_For\_Multi\_Robots$(P, T, d)$} at robot $r_i$}
\label{alg:SAdjust}
\begin{tabbing}
xxx \= xxx \= xxx \= xxx \= xxx \= xxx \= xxx \= xxx \= xxx \= xxx
\kill 
{\bf Input at robot $r_i$} \\ 
\> $P$: current configuration \\ 
\> $T$: target pattern \\ 
\> $d$: size of the target pattern \\
{\bf Notations} \\ 
\> $P_1, P_2, \ldots, P_{n/\rho(P)}$: $\rho(P)$-decomposition of $P$ \\ 
\> $c'(P)$: single recognized center in $P$ if any \\ 
\> $P'_1, P'_2, \ldots, P'_k$: $\rho^*(P,c'(P))$-decomposition of $P$\\ 
\> $SEC'(P)$: the smallest enclosing circle of $P$ centered at $c'(P)$ \\ 
\> $LEC'(P)$: circle centered at $c'(P)$ and containing $P_1'$ \\ 
\> $d^*$: Minimum distance to make a center recognized \\ 
{\bf Movement rule} \\ 
\> All movements are along $r_i$'s radius toward $c(P)$ unless specified otherwise \\
\\
{\bf Algorithm} \\ 
\> {\bf If $rad(P) = d/2$} {\bf then} terminate \\ 
\> {\bf Else} \\ 
\> \> {\bf If} a single center is not recognized in $P$ {\bf then} // Step 1 \\ 
\> \> \> {\bf If} $\rho(P) = n$ {\bf then} move to the point at distance $d/2$ from $c(P)$\\ 
\> \> \> {\bf If} $\rho(P) = 1$ {\bf then} \\ 
\> \> \> \> Let $P_i$ and $P_j$ ($i < j$) be the elements with the minimum indices \\ 
\> \> \> \> of the $\rho(P)$-decomposition of $P$ satisfying $SEC(P) = SEC(P \setminus (P_i \cup P_j))$. \\
\> \> \> \> {\bf If} $r_i \in P_i$ {\bf then} move toward $c(P)$ but stop before reaching $c(P)$\\ 
\> \> \> \> {\bf Else if} $r_i \in P_j$ {\bf then} \\ 
\> \> \> \> \> move to a position at distance $d^*$ from $P_1$ \\ 
\> \> \> {\bf If}  ($\rho(P) \neq 1, n$) and ($r_i \in P_1$) {\bf then} \\ 
\> \> \> \> move to a position at distance $d^*/2$ from $c(P)$\\ 
\> \> {\bf Else} // $c'(P)$ is recognized \\ 
\> \> \> {\bf If} (radius of $SEC'(P)$ is larger than $d/2$) and \\ 
\> \> \> \> ($r_i$ is not on its designated position) {\bf then} // Step 2 \\ 
\> \> \> \> {\bf If} (($\rho^*(P,c'(P)) > 1$) and ($r_i \in P \setminus (P_1' \cup P_k')$)) or \\ 
\> \> \> \> \> (($\rho^*(P,c'(P)) = 1$) and ($r_i \in P \setminus (P_1' \cup P_2' \cup P_{k-2}' \cup P_{k-1'} \cup P_k')$)) {\bf then} \\ 
\> \> \> \> \> Move cautiously toward $r_i$'s designated position \\ 
\> \> \> \> \> but stop before creating a new minimum distance pair\\ 
\> \> \> {\bf Else} // Step 3 \\ 
\> \> \> \> {\bf If} (radius of $SEC'(P)$ is not $d/2$) {\bf then} \\
\> \> \> \> \> {\bf If} (($\rho^*(P,c'(P)) > 1$) and ($r_i \in P_k'$)) or \\   
\> \> \> \> \> \> (($\rho^*(P,c'(P)) = 1$) and ($r_i \in P_{k-2}' \cup P_{k-1'} \cup P_k')$)) {\bf then} \\   
\> \> \> \> \> \> Move to the position at distance $d/2$ from $c'(P)$\\ 
\> \> \> \> {\bf Else} // $\rho^*(P,c'(P)) = 1$ and $rad(P) \neq d/2$ \\ 
\> \> \> \> \> // i.e., $SEC'(P)$ has an empty half arc \\  
\> \> \> \> \> {\bf If} $r_i$ is the counterclockwise endpoint \\ 
\> \> \> \> \> of the longest empty arc of $SEC'(P)$ {\bf then} \\ 
\> \> \> \> \> \> Move counterclockwise along $SEC'(P)$ \\ 
\> \> \> \> \> \> until it form a diameter with another robot on $SEC'(P)$ 
\end{tabbing}
\end{algorithm}

\subsection{Correctness}

Let ${\mathcal C}_n$ be the set of all configurations of the $n$ robots. 
Then, we consider the decomposition of ${\mathcal C}_n$ into the 
terminal configurations of the three steps. 
\begin{itemize}
    \item ${\mathcal C}_{d} \in {\mathcal C}_n$ be the set of all configurations $P$ of the $n$ robots, where the diameter of $SEC(P)$ is $d$. 
    \item ${\mathcal C}_{n,\text{center}} \subseteq {\mathcal C}_n \setminus {\mathcal C}_{d}$ 
        be the set of all configurations of the $n$ robots, where the center is recognized. 
    \item ${\mathcal C}_{n,\text{aligned}} \subseteq {\mathcal C}_{n, center}$ 
        be the set of all configurations of the $n$ robots, 
        where the robots are on their designated positions. 
\end{itemize}

\begin{lemma}
\label{lemma:center}
In any execution $P(0), P(1), P(2),\ldots$ of Algorithm~\ref{alg:SAdjust} 
starting from an initial configuration $P(0) \in {\mathcal C}_n \setminus {\mathcal C}_{d}$, 
there exists finite $t$ such that $P(t) \in {\mathcal C}_{n, \text{center}}$, 
$P(t)$ does not contain any multiplicity, and $\rho(P(0)) \geq \rho(P(t))$. 
\end{lemma}
\begin{proof}
We consider a configuration $P \not\in {\mathcal C}_{n, \text{center}}$. 
The robots can agree on the total ordering among the elements 
of the $\rho(P)$-decomposition of $P$, say $\{P_1, P_2, \ldots, P_{n/\rho(P)}\}$. 
Algorithm~\ref{alg:SAdjust} sends $P_1$ to the interior of 
$LEC(P)$ so that their destinations form a new minimum distance. 
In the SSYNC model, at least one robot moves toward $c(P)$. 
If $\rho(P) < n$, $P_{n/\rho(P)}$ does not move 
and in the resulting configuration $P'$, $c(P') = c(P)$. 
Thus, the distance to $c(P')$ and the moving robot becomes smaller than that in $P$. 
During the transition from $P$ to $P'$, $\rho(P')$ may become smaller 
due to SSYNC activation and non-rigid movement. 
However, the first element of the $\rho(P')$-decomposition of $P'$ 
moves toward $c(P')$ to form a new minimum distance. 
In this way, Step 1 of Algorithm~\ref{alg:SAdjust} is repeated until 
the robots reach some configuration in $ {\mathcal C}_{n, \text{center}}$. 
The movements of Step 1 does not create any multiplicity 
because for each radius of $SEC(P)$, at most one robot nearest to $c(P)$ 
moves toward $c(P)$. 

If $\rho(P)=n$, we have the following two cases. 
If the $SEC(P')$ of the resulting configuration $P'$ is the same as $SEC(P)$, 
the distance between $c(P')(=c(P))$ and the moving robot becomes smaller. 
Otherwise, $SEC(P')$ is contained in $SEC(P)$, 
and the distance between $c(P')$ and the moving robot becomes smaller. 
In the same way as the previous case,  
the robots eventually reach some configuration in $ {\mathcal C}_{n, \text{center}}$. 
\qed
\end{proof}

\begin{lemma}
\label{lemma:align}
We consider an initial configuration $P(0) \in {\mathcal C}_{n, \text{center}}$
such that $rad(P(0))$ is larger than $d/2$. 
In any execution $P(0), P(1), P(2),\ldots$ of Algorithm~\ref{alg:SAdjust}, 
there exists finite $t$ such that $P(t) \in {\mathcal C}_{n,\text{aligned}}$ 
and $\rho(P(0)) \geq \rho(P(t))$. 
\end{lemma} 
\begin{proof}
Algorithm~\ref{alg:SAdjust} sends the robots to their designated positions 
by cautious walk so that the robots do not create a new minimum distance. 
If no robot can move in $P(t) \not\in {\mathcal C}_{n,\text{aligned}}$, 
there exists at least one pair of robots on the same radius 
heading to opposite directions. 
This is a contradiction because assignment of the designated positions 
does not contain any such intersections. 
Hence, the designated positions are gradually occupied and 
the robots eventually reach a configuration in ${\mathcal C}_{n,\text{aligned}}$. 

During the execution, the symmetricity among the robots does not increase 
because $P(0) \in {\mathcal C}_{n, \text{center}}$ and 
the robots nearest to the recognized center keeps the symmetricity. 
\qed
\end{proof}

\begin{lemma}
\label{lemma:adjust}
In any execution $P(0), P(1), P(2),\ldots$ of Algorithm~\ref{alg:SAdjust} 
there exists finite $t$ where $P(t) \in {\mathcal C}_{d}$. 
\end{lemma}
\begin{proof}
We first consider the case where the size of $P(0)$ is larger than $d$. 
By Lemma~\ref{lemma:center} and \ref{lemma:align}, 
the robots eventually reach a configuration $P(t)$ in ${\mathcal C}_{n,\text{aligned}}$. 
In $P(t)$, the robots on $SEC(P(t))$ move toward the recognized center. 
By $P(t)$ in ${\mathcal C}_{n,\text{aligned}}$, there is no other robot 
on the trajectory, and this movement does not create a new minimum distance. 
Hence, in $P(t+1)$, the recognized center is the same, 
and the robots at distance larger than $d/2$ from the recognized center 
move toward the center. 
In this way, Step 3 of Algorithm~\ref{alg:SAdjust} is repeated until 
the size of the current configuration becomes $d$. 

When the size of $P(0)$ is smaller than $d$, 
we have the same discussion. 
\qed
\end{proof}

Consequently, we have the following theorem. 
\begin{theorem}
\label{theorem:adjust}
More than four oblivious SSYNC robots can translate an initial configuration $P$ 
with $rad(P) \neq d/2$ into another configuration $P'$ with 
$rad(P') = d/2$ and $\rho(P') \leq \rho(P)$ by Algorithm~\ref{alg:SAdjust}. 
\end{theorem}

The size-adjusting algorithm 
allows the robots on the SEC centered at 
the recognized center to move in the final step. 
Hence, we have the following Corollary. 
\begin{corollary}
The size-adjusting algorithm 
terminates  
as soon as the size of SEC of the robots becomes the specified size $d$.  
\end{corollary}

Fujinaga et al. proposed a pattern formation algorithm for oblivious ASYNC robots~\cite{FYOKY15}, 
that does not change the SEC during any execution. 
We can obtain a \HwPTF \ algorithm 
by combining Algorithm~\ref{alg:SAdjust} with the pattern formation algorithm in \cite{FYOKY15}:  
the robots switch between the two algorithms based on the size of their SEC, i.e., 
if their SEC is larger than the size of the target pattern, 
each robot executes our size adjusting algorithm,  
otherwise each robot executes the pattern formation algorithm. 

\begin{theorem}
When $n (\geq 5)$ oblivious SSYNC robots agree on the unit distance, 
they can solve the \HwPTF \ problem 
for a target pattern $F$ 
from any initial configuration $I$ if and only if $\rho(I)$ divides $\rho(F)$. 
\end{theorem}

\section{Conclusion}
In this paper, we newly introduced the \HwPTF \ problem, 
that requires the robots to form a target pattern in a specified size. 
We started with the \HwPTF \ problem for two robots and 
showed the problem is solvable for oblivious SSYNC robots 
when the specified size is not zero. 
We then showed the problem is not solvable for oblivious ASYNC robots 
while it is solvable by oblivious ASYNC robots equipped with lights.  
Finally, we presented a size-adjusting algorithm for the oblivious 
SSYNC robots and showed that we can obtain a \HwPTF \ algorithm 
by combining it with an existing pattern formation algorithm. 

There are interesting future directions. 
First, the \HwPTF \ problem for more than two oblivious ASYNC robots is open. 
We believe that we can extend the impossibility for two ASYNC robots 
for initial configurations where $n$ ($n >2$) ASYNC robots 
forming a regular $n$-gon. 
On the other hand, size-adjusting could be possible 
for other initial configurations. 
Another direction is to consider \HwPTF \ problem in the 3D space; obviously this setting would require considering 3D rotational symmetry~\cite{YUKY17}. 

Finally, all results in this paper consider deterministic algorithms. We conjecture that our impossibility results for two deterministic ASYNC robots can be extended to the case of two probabilistic ASYNC robots.

\subsubsection*{Acknowledgements}
We thank Nicola Santoro and Paola Flocchini 
for originally proposing the hardwired pattern formation problem studied in this paper.

%
%
%
\bibliographystyle{splncs04}
\bibliography{bibitems}

@article{Courtieu15,
  author       = {Pierre Courtieu and
                  Lionel Rieg and
                  S{\'{e}}bastien Tixeuil and
                  Xavier Urbain},
  title        = {Impossibility of gathering, a certification},
  journal      = {Inf. Process. Lett.},
  volume       = {115},
  number       = {3},
  pages        = {447--452},
  year         = {2015},
  url          = {https://doi.org/10.1016/j.ipl.2014.11.001},
  doi          = {10.1016/J.IPL.2014.11.001},
  bibsource    = {dblp computer science bibliography, https://dblp.org}
}

@article{Dijkstra74,
  author       = {Edsger W. Dijkstra},
  title        = {Self-stabilizing Systems in Spite of Distributed Control},
  journal      = {Commun. {ACM}},
  volume       = {17},
  number       = {11},
  pages        = {643--644},
  year         = {1974},
  url          = {https://doi.org/10.1145/361179.361202},
  doi          = {10.1145/361179.361202},
  bibsource    = {dblp computer science bibliography, https://dblp.org}
}

@article{Bramas23,
  author       = {Quentin Bramas and
                  Anissa Lamani and
                  S{\'{e}}bastien Tixeuil},
  title        = {The agreement power of disagreement},
  journal      = {Theor. Comput. Sci.},
  volume       = {954},
  pages        = {113772},
  year         = {2023},
  url          = {https://doi.org/10.1016/j.tcs.2023.113772},
  doi          = {10.1016/J.TCS.2023.113772},
  bibsource    = {dblp computer science bibliography, https://dblp.org}
}

@inproceedings{Viglietta13,
  author       = {Giovanni Viglietta},
  title        = {Rendezvous of Two Robots with Visible Bits},
  booktitle    = {Proceedings of the 9th International Symposium on Algorithms
                  and Experiments for Sensor Systems, Wireless Networks and Distributed
                  Robotics (ALGOSENSORS)},
  pages        = {291--306},
  year         = {2013},
  doi          = {10.1007/978-3-642-45346-5\_21},
}

@article{Her2Cols,
  author       = {Xavier D{\'{e}}fago and
                  Adam Heriban and
                  S{\'{e}}bastien Tixeuil and
                  Koichi Wada},
  title        = {Using model checking to formally verify rendezvous algorithms for
                  robots with lights in Euclidean space},
  journal      = {Robotics Auton. Syst.},
  volume       = {163},
  pages        = {104378},
  year         = {2023},
  url          = {https://doi.org/10.1016/j.robot.2023.104378},
  doi          = {10.1016/J.ROBOT.2023.104378},
  bibsource    = {dblp computer science bibliography, https://dblp.org}
}

@article{CFPS12, 
author="Mark Cieliebak and Paola Flocchini and Giuseppe Prencipe and Nicola Santoro", 
title="Distributed computing by mobile robots: gathering", 
journal="SIAM J. Comput.",
volume="41",
number="4",
pages="829--879",
year="2012",  
doi="10.1137/100796534"
}

@article{DFPSY16,
author="Shantanu Das and Paola Flocchini and Giuseppe Prencipe and 
Nicola Santoro and Masafumi Yamashita",  
title="Autonomous mobile robots with lights",
journal="Theor. Comput. Sci",
volume="609",
pages="171--184",
year="2016",
doi="10.1016/j.tcs.2015.09.018"
}

@InProceedings{DPV10, 
author="Dieudonn{\'e}, Yoann
and Petit, Franck
and Villain, Vincent",
title="Leader Election Problem versus Pattern Formation Problem",
booktitle="Proceedings of the 24th International Symposium on Distributed Computing
Distributed Computing",
year="2010",
pages="267--281",
doi="10.1007/978-3-642-15763-9_26"
}

@article{FPSW08,
author="Paola Flocchini and Giuseppe Prencipe and 
Nicola Santoro and Peter Widmayer", 
title="Arbitrary pattern formation by asynchronous, anonymous, oblivious 
robots",
journal="Theor. Comput. Sci.",
volume="407",
pages="412--447",
year="2008",
doi="10.1016/j.tcs.2008.07.026"
}

@article{FYOKY15,
author="Nao Fujinaga and Yukiko Yamauchi and Hirotaka Ono and 
Shuji Kijima and Masafumi Yamashita", 
title="Pattern formation by oblivious asynchronous mobile robots",
journal="SIAM J. Comput.",
volume="44",
number="3",
pages="740--785",
year="2015", 
doi="10.1137/140958682"
}

@InProceedings{HHK24,
author={Hahn, Christopher and Harbig, Jonas and Kling, Peter},
title={{Forming Large Patterns with Local Robots in the OBLOT Model}},
booktitle={Proceedings of the 3rd Symposium on Algorithmic Foundations of Dynamic Networks (SAND 2024)},
pages={14:1--14:20},
year={2024},
doi={10.4230/LIPIcs.SAND.2024.14}
}

@article{SY99,  
author="Ichiro Suzuki and Masafumi Yamashita", 
title="Distributed anonymous mobile robots: Formation of geometric patterns", 
journal="SIAM J. Comput.", 
volume="28", 
number="4", 
pages="1347--1363", 
year="1999", 
doi="10.1137/S009753979628292X"
}

@article{YS10, 
author="Masafumi Yamashita and Ichiro Suzuki", 
title="Characterizing geometric patterns formable by oblivious anonymous mobile robots", 
journal="Theor. Comput. Sci.", 
volume="411", 
pages="2433--2453", 
year="2010", 
doi="10.1016/j.tcs.2010.01.037"
}

@article{YUKY17,
author="Yukiko Yamauchi and Taichi Uehara and Shuji Kijima and Masafumi Yamashita",
title="Plane Formation by Synchronous Mobile Robots in the Three Dimensional Euclidean Space", 
journal="Journal of the ACM", 
volume="64",
issue="3",
pages="16:1--16:43", 
year="2017",
doi="10.1145/3060272"
}

@inproceedings{YY13, 
author="Yukiko Yamauchi and Masafumi Yamashita", 
title="Pattern formation by mobile robots with limited visibility",
booktitle="Proceedings of the 20th International Colloquium on
Structural Information and Communication Complexity (SIROCCO 2013)", 
publisher="Springer International Publishing", 
pages="201--212",
year="2013", 
doi="10.1007/978-3-319-03578-9_17"
}
\end{document}